\documentclass[12pt]{article}

\usepackage{epsf,epsfig,amsmath,amsfonts,amsthm}
\usepackage{graphicx}
\usepackage{color}
\usepackage{accents}
\usepackage{mathrsfs}
\usepackage{tikz}
\usepackage{hyperref}

\newtheorem{theorem}{Theorem}

\newtheorem{proposition}[theorem]{Proposition}
\newtheorem{lemma}[theorem]{Lemma}
\newtheorem{definition}{Definition}
\newtheorem{remark}{Remark}

\usepackage{graphics} 
  \usepackage{epsfig}
\usepackage{lipsum}
\usepackage{amsfonts}
\usepackage{amssymb}
\usepackage{graphicx}
\usepackage{epstopdf}
\usepackage{algorithm}
\usepackage{algorithmic}
\usepackage{graphicx}
\usepackage{enumitem}
\usepackage{tikz}
\usepackage{booktabs}
\usepackage{float}
\usetikzlibrary{positioning,arrows.meta,shadows,fit}

\usepackage{subfigure}
\ifpdf
  \DeclareGraphicsExtensions{.eps,.pdf,.png,.jpg}
\else
  \DeclareGraphicsExtensions{.eps}
\fi

\usepackage{enumitem}

\usepackage{subfigure}
\DeclareMathOperator{\sign}{sign}

\newcommand{\ts}{&\hspace{-0.08in}}
\newcommand{\nn}{\nonumber}
\newcommand{\beq}{\begin{equation}}
\newcommand{\eeq}{\end{equation}}
\newcommand{\bea}{\begin{eqnarray}}
\newcommand{\eea}{\end{eqnarray}}
\newcommand{\bef}{\begin{flalign}}
\newcommand{\eef}{\end{flalign}}
\usepackage{tikz}
\usetikzlibrary{positioning, arrows.meta, shapes, fit, shapes.misc, calc}

\begin{document}

\title{Recursive Binary Identification under Data Tampering and Non-Persistent Excitation with Application to Emission Control}

\author{
{Jian Guo\footnote{Department of Applied Mathematics, The Hong Kong Polytechnic University, Kowloon, Hong Kong (\texttt{j.guo@amss.ac.cn}).},
\ Lihong Pei\footnote{State Key Laboratory of Mathematical Sciences, 
Academy of Mathematics and Systems Science, Chinese Academy of Sciences, Beijing 100190, China, 
and School of Mathematical Sciences, University of Chinese Academy of Sciences, Beijing 100049, China 
(\texttt{plh1225@amss.ac.cn}, \texttt{wenchaoxue@amss.ac.cn}, \texttt{ylzhao@amss.ac.cn}).},
\ Wenchao Xue\footnotemark[2],
\ Yanlong Zhao\footnotemark[2],
\ and Ji\mbox{-}Feng Zhang\footnote{School of Automation and Electrical Engineering, 
Zhongyuan University of Technology, Zhengzhou 450007, China; 
State Key Laboratory of Mathematical Sciences, Academy of Mathematics and Systems Science, 
Chinese Academy of Sciences, Beijing 100190, China; 
and School of Mathematical Sciences, University of Chinese Academy of Sciences, Beijing 100049, China 
(\texttt{jif@iss.ac.cn}).}
}}

\maketitle
\begin{abstract} \noindent
This paper studies the problem of online parameter estimation for cyber-physical systems with binary outputs that may be subject to adversarial data tampering. Existing methods are primarily offline and unsuitable for real-time learning. To address this issue, we first develop a first-order gradient-based algorithm that updates parameter estimates recursively using incoming data. Considering that persistent excitation (PE) conditions are difficult to satisfy in feedback control scenarios, a second-order quasi-Newton algorithm is proposed to achieve faster convergence without requiring the PE condition. For both algorithms, corresponding versions are developed to handle known and unknown tampering strategies, and their parameter estimates are proven to converge almost surely over time. In particular, the second-order algorithm ensures convergence under a signal condition that matches the minimal excitation required by classical least-squares estimation in stochastic regression models. The second-order algorithm is also extended to an adaptive control framework, providing an explicit upper bound on the tracking error for binary-output FIR systems under unknown tampering. Three numerical simulations verify the theoretical results and show that the proposed methods are robust against data tampering. Finally, the approach is validated via a vehicle emission control problem, where it effectively improves the detection accuracy of excess-emission events.
\end{abstract}

\section{Introduction}
Cyber–Physical Systems (CPS) connect people, machines, environments, and computational components, enabling close interaction between the physical and digital worlds~\cite{kazemi2021finite,ma2024ctomp}. By combining sensing, communication, computing, and control, CPS can respond to real-time changes, adapt to dynamic environments, and continuously improve system performance. These features make CPS a core technology in many industrial applications, such as automated manufacturing, intelligent transportation, healthcare, energy systems, and smart infrastructure~\cite{lee2015cyber,hahn2013cyber}.

Despite their advantages, CPS often face cybersecurity risks. These vulnerabilities mainly stem from their distributed structures and use of inexpensive sensors and communication devices that frequently lack strong protection~\cite{peng2017optimal,Pasqualetti2013,Kwon2018}. Such weaknesses allow attackers to disrupt operations, making security crucial in CPS design~\cite{fawzi2014secure}. Among various cyber threats, Denial-of-Service (DoS) and Data Tampering attacks are particularly common and harmful~\cite{peng2017optimal}. DoS attacks block or delay data transmission~\cite{pelechrinis2010denial}, while Data Tampering attacks subtly alter transmitted data, making detection difficult~\cite{jovanov2019relaxing}. These attacks can lead to wrong decisions, reduced performance, or system failures, potentially causing production delays, financial losses, and safety incidents~\cite{Oliva2018,Yamada2019}.

Several methods have been proposed to strengthen CPS security. For instance, Bayesian-based defense methods detect deceptive attacks by continuously updating the system's belief about potential threats~\cite{sasahara2023asymptotic,umsonst2023bayesian}. Other studies have developed proactive schemes to protect critical infrastructures, such as power grids, against dynamic attacks~\cite{amini2016dynamic}. Research on data-injection attacks has considered trade-offs between attack detection and control performance~\cite{bai2017data}, and also examined balancing security with privacy in connected systems~\cite{katewa2021security}. Adaptive control techniques have been widely explored to maintain stable system operation even if sensors or actuators are compromised~\cite{jin2017adaptive,yucelen2016adaptive,haddad2025adaptive}. Additionally, secure-by-construction approaches have been developed to design CPS with built-in security measures, reducing vulnerabilities from the start~\cite{liu2022secure,zhong2025secure}.

In addition to security risks, CPS also commonly face challenges related to quantization. Due to cost constraints and limited bandwidth, sensors in many CPS often send quantized measurements rather than exact values. Binary outputs from switching sensors or optical detectors are typical examples~\cite{Wang2010}. While quantization helps reduce data size and bandwidth usage, it also introduces nonlinearities that complicate system identification and control~\cite{Wu2013}. To address this challenge, researchers have developed a variety of identification methods tailored for quantized systems~\cite{Wang2003}, including studies on binary sensor applications~\cite{Guo2015,Bottegal2017,Pouliquen2020}, worst-case estimation techniques~\cite{Casini2011}, and methods aimed at achieving asymptotic efficiency~\cite{wang2024asymptotically}.

Beyond quantization, Data Tampering attacks also further complicate the identification problem in CPS. Researchers have proposed various methods to address such attacks in binary-output systems. For instance, \cite{guo2020system} developed a compensation-based algorithm that ensures strong consistency and asymptotic normality under tampering, later extending it to multi-dimensional systems~\cite{guo2023identification}. To handle packet loss during transmission, \cite{guo2021system} proposed algorithms for both known and unknown loss rates, integrating compensation mechanisms to reduce communication costs. More recently, \cite{guo2023iterative} 
proposed a maximum likelihood estimation (MLE) approach using an iterative Expectation-Maximization algorithm for parameter identification under tampering. Furthermore, \cite{guo2025identification} explored finite impulse response (FIR) systems under random replay attacks, presenting consistent estimation algorithms and asymptotically optimal defense strategies.

Despite these advances, many existing methods rely heavily on offline processes such as empirical measure estimation or MLE, which may not be suitable for real-time CPS that continuously generate data. This gap highlights the need for online identification algorithms that can update estimates dynamically while maintaining robustness against tampering. To address this challenge, this paper extends our previous work~\cite{your_conference_paper} by presenting a unified identification framework. We propose both first-order and second-order recursive algorithms that handle cases with known attack strategies and scenarios requiring online estimation of tampering statistics.

The main contributions of this paper are summarized as follows:

\begin{itemize}
    \item This paper presents a first-order gradient-based algorithm for parameter estimation under data tampering attacks. The convergence properties of the algorithm, including almost sure convergence and mean-square convergence, are established. To handle more realistic scenarios, a periodic extra-insertion defense mechanism is introduced to estimate tampering statistics during the identification process and ensure convergence under unknown attacks.

    \item A projected second-order quasi-Newton algorithm is proposed to handle both known and unknown tampering strategies, achieving faster convergence. Integrated into an adaptive control framework, the algorithm provides explicit asymptotic bounds on the estimation error, tracking error, and cumulative regret. These bounds are established under the weakest excitation condition in~\cite{lai1982least}, without requiring persistent excitation (PE) condition.
    
\item The proposed methods are applied to vehicle emission control, using a dataset of diesel vehicle emissions collected in Hefei, China. Transformer-based features are used as regressors in the set-valued model, where “over-limit” and “non-over-limit” cases are modeled as binary outcomes. Data tampering, such as underreporting excessive emissions, is captured via probabilistic label manipulation. The results show improved accuracy in identifying compliant and non-compliant vehicles, even under adversarial conditions.

\end{itemize}

The paper is organized as follows. Section 2 introduces the system and examines identifiability under tampering attacks. Section 3 presents a first-order gradient-based identification algorithm, establishes its convergence, and proposes a periodic insertion scheme to handle unknown attack strategies. Section 4 develops a second-order Newton-type algorithm and provides theoretical bounds on estimation error, cumulative regret, and tracking performance. Section 5 offers numerical simulations and a real-world case study on vehicle emission control. Section 6 concludes the paper and outlines future directions.


\textbf{Notation:} Let \( \mathbb{Z}_+ \) be the set of positive integers. For a positive integer \( k \), define \( [k] = \{1, 2, \dots, k\} \). For a set \( \mathcal{S} \subset \mathbb{Z}_+ \) and an integer \( l \), define \( l\mathcal{S} = \{i : i = l \times j, j \in \mathcal{S} \} \). For a sequence of sets \( \{\mathcal{S}_k\}_{k\geq 0} \), if \( l > j \), then \( \bigcup_{k=l}^{j} \mathcal{S}_k = \emptyset \).

\section{Model Formulation and Identifiability}
\subsection{Model Description Under Tampering Attack}

This subsection introduces the identification problem of stochastic FIR systems under tampered binary observations. Consider the following system:
\begin{equation} \label{s1}
y_{k+1} = \varphi_k^T \theta + w_{k+1}, \quad k = 0, 1, \ldots,
\end{equation}
where $\theta \in \Theta\subseteq \mathbb{R}^p$ is an unknown but time-invariant parameter vector of known dimension $p$, $\varphi_k \in \mathbb{R}^p$ is the regressor vector composed of current and past input signals, and $w_{k+1}$ is a stochastic noise sequence with zero mean. Let $\mathcal{F}_k$ denote the natural filtration defined as
\begin{equation} \label{filtration}
\mathcal{F}_k \triangleq \sigma\left\{ \varphi_0, \dots, \varphi_k, w_0, \dots, w_k \right\}, \quad k \geq 0.
\end{equation}
In this setting, the system output $y_{k+1}$ is not directly observable. Instead, it is accessed through a binary sensor with a known threshold $C \in \mathbb{R}$, which generates the binary-valued signal:
\begin{equation} \label{s2}
s_{k+1}^0 = I_{[y_{k+1} \leq C]} = 
\begin{cases}
1, & y_{k+1} \leq C, \\
0, & \text{otherwise}.
\end{cases}
\end{equation}
However, this binary signal is transmitted over a potentially compromised communication channel. Due to data tampering attacks, the received signal $s_{k+1}$ may differ from the original $s_{k+1}^0$, and is modeled by the following flipping probabilities:
\begin{equation} \label{s3}
\left\{
\begin{aligned}
\Pr\{ s_k = 0 \mid s_k^0 = 1 \} &= p, \\
\Pr\{ s_k = 1 \mid s_k^0 = 0 \} &= q,
\end{aligned}
\right.
\end{equation}
where $p, q \in [0, 1)$ represent the attacker's flipping strategy. 

\medskip
\noindent

\textbf{Identification Objective.} The objective of this paper is to design a recursive algorithm to estimate the unknown parameter $\theta$ based on the sequence of regressors $\{ \varphi_k \}_{k\geq 0}$ and the possibly tampered binary observations $\{ s_k \}_{k\geq}$. Further, if control actions are required, the estimation result $\hat{\theta}$ from the estimation center is transmitted to the control center. The control center then utilizes this estimated parameter to design the controller input $u_k$, forming a closed-loop feedback system. The flow of information and control within this closed-loop system, incorporating potential data tampering attacks, is illustrated in Fig.~\ref{fig:closed-loop}.


\begin{remark}
In many secure control architectures, the communication network is logically divided into an “downlink” (from the sensor to the estimation center) and a “uplink” (from the control center to the actuator/execution unit). By implementing the uplink over a protected, dedicated channel—such as a private optical‐fiber link, a local bus, or an encrypted virtual tunnel—one can safely assume that the control signal \(u_k\) is not subject to tampering during transmission. As a result, the integrity risk is primarily concentrated on the downlink path that conveys the binary sensor measurements \(s_k^0\) through a potentially untrusted network \cite{zheng2022survey,lee2021downlink}.
\end{remark}

\begin{figure}[t]
    \centering
    \includegraphics[width=0.8\textwidth]{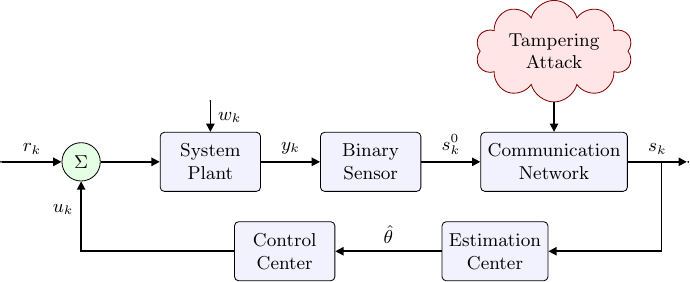}
    \caption{Closed‐loop control system flowchart with tampering attack and estimation/control centers}
    \label{fig:closed-loop}
\end{figure}

\subsection{Identifiability Analysis}
Let $F_k(\cdot)$ denote the conditional distribution function of the noise $w_{k+1}$ given the filtration $\mathcal{F}_k$. Based on the law of total probability and the tampering model in~\eqref{s3}, the conditional probability of observing $s_{k+1} = 0$ can be computed as
\begin{align}
&P(s_{k+1} = 0 \mid \mathcal{F}_k)\notag\\
=& P(s_{k+1}^0 = 1 \mid \mathcal{F}_k) \cdot P(s_{k+1} = 0 \mid s_{k+1}^0 = 1, \mathcal{F}_k) \notag \\
&+ P(s_{k+1}^0 = 0 \mid \mathcal{F}_k) \cdot P(s_{k+1} = 0 \mid s_{k+1}^0 = 0, \mathcal{F}_k) \notag \\
=& p F_k(C - \theta^T \varphi_k) + (1 - q)\left[1 - F_k(C - \theta^T \varphi_k)\right] \notag \\
=& (p + q - 1) F_k(C - \theta^T \varphi_k) + 1 - q. \label{dist1}
\end{align}
Similarly, the conditional probability of $s_{k+1} = 1$ is given by
\begin{align}
P(s_{k+1} = 1 \mid \mathcal{F}_k) 
&= 1 - P(s_{k+1} = 0 \mid \mathcal{F}_k) \notag \\
&= (1 - (p + q)) F_k(C - \theta^T \varphi_k) + q. \label{dist2}
\end{align}

\medskip
\noindent
\textbf{Identifiability.} This subsection discusses the identifiability of the binary-output system defined in \eqref{s1}--\eqref{s3}. From equations~\eqref{dist1}--\eqref{dist2}, it can be observed that when $p + q = 1$, the dependence on the conditional distribution $F_k(\cdot)$ vanishes. In this degenerate case, the distribution of the observations becomes independent of $\theta$, making parameter identification impossible. Therefore, a necessary condition for identifiability is:
$
p + q \neq 1.
$

In addition, to recover the unknown parameter $\theta$ from the sequence of binary observations, the regressor sequence $\{\varphi_k\}$ must be sufficiently informative \cite{ljung1995system}. A formal excitation condition will be specified in subsequent sections, depending on the structure of the identification algorithm.

\section{First-order Gradient Identification Algorithm}
\subsection{Known Attack Strategy}

A first-order recursive identification algorithm for FIR systems under binary observations with data tampering was introduced in earlier work~\cite{your_conference_paper}. This subsection briefly reviews the algorithm, including its motivation, formulation, and main convergence results. Complete proofs are presented in the Appendix.

\subsubsection{Motivation and Algorithm}

Classical identification methods, such as stochastic approximation (SA), yield consistent parameter estimates as the data length increases. However, when only binary observations are available and such data may be tampered, traditional approaches are no longer applicable. To compensate for the limited information, additional assumptions-such as knowledge of the noise distribution-are required.

The key idea is based on the fact that, in system identification, the conditional variance of $w_{k+1}$ is minimized when using the true parameter $\theta$. For the system~\eqref{s1}, we have:
\bea
\mathrm{Var}(w_{k+1}|\mathcal{F}_k)\ts =\ts \mathbb{E}\left[(y_{k+1} - \theta^T \varphi_k)^2 \mid \mathcal{F}_k\right]\nn \\
\ts\leq\ts \mathbb{E}\left[(y_{k+1} - \beta^T \varphi_k)^2 \mid \mathcal{F}_k\right],\ \forall \beta \in \Theta.\nn
\eea
Under appropriate excitation conditions on $\{\varphi_k\}$, the equality can be held if and only if $\beta = \theta$. This observation allows us to reformulate the identification problem as the following stochastic optimization:
\begin{equation}\label{loss1}
\min_{\beta \in \Theta} I(\beta) := \tfrac{1}{2} \mathbb{E} \left[(y_{k+1} - \beta^T \varphi_k)^2 \mid \mathcal{F}_k \right].
\end{equation}

To solve the above problem, stochastic gradient descent (SGD) can be applied by using a single sample at each step. Specifically, the parameter is updated based on the observed pair $(y_{k+1}, \varphi_k)$ as:
\bea
\theta_{k+1} \ts=\ts \theta_k - b_k \nabla_\beta I(\beta, y_{k+1}, \varphi_k) \big|_{\beta = \theta_k} \nn\\
\ts=\ts \theta_k + b_k (y_{k+1} - \theta_k^T \varphi_k)\varphi_k,\nn
\eea
where $\{b_k\}$ is a sequence of diminishing step sizes. This recursion corresponds to the classical first-order stochastic approximation algorithm, which is widely used in system identification~\cite{chen2021stochastic}.

In the presence of tampered binary outputs, as described by the system~\eqref{s1} with binary observations~\eqref{s2} and potential tampering~\eqref{s3}, a similar strategy can be employed to develop a recursive identification algorithm. At each time step $k$, the conditional expectation of the observed binary output $s_{k+1}$ is given by
\begin{equation} \label{mean1}
\mathbb{E}(s_{k+1} \mid \mathcal{F}_k) = (1 - (p + q)) F_k(C - \theta^T \varphi_k) + q.
\end{equation}
As in the case of $y_{k+1}$, the conditional variance of $s_{k+1}$ is minimized when $\beta = \theta$. This also allows the identification problem to be formulated as a stochastic optimization task, using a tampered binary loss function:
\bea 
\ts\ts\min_{\beta \in \Theta} I_{TB}(\beta, s_{k+1}, \varphi_k) \label{loss2}\\
\ts:=\ts \tfrac{1}{2} \mathbb{E} \left( \left( s_{k+1} - \left[(1 - (p + q)) F_k(C - \beta^T \varphi_k) + q \right] \right)^2 \mid \mathcal{F}_k \right).\nn
\eea
The gradient of the loss function with respect to $\beta$ is
\bea
\ts\ts\nabla_\beta I_{TB}(\beta, s_{k+1}, \varphi_k) = - (1 - (p + q))f_k(C - \beta^T \varphi_k)\nn\\
\ts\ts\times\left( (1 - (p + q)) F_k(C - \beta^T \varphi_k) + q - s_{k+1} \right) \varphi_k.\nn
\eea

To ensure that $f_k(C - \beta^T \varphi_k)$ does not approach zero, it is natural to constrain $\beta$ within a compact parameter set $\Theta$. This motivates the use of a projection operator, defined as follows.

\begin{definition}\label{def1}
Let $\Omega \subseteq \mathbb{R}^n$ be a convex compact set. The projection operator $\Pi_{\Omega}(\cdot)$ is defined by
\[
\Pi_{\Omega}(x) = \arg\min_{\omega \in \Omega} \|x - \omega\|, \quad \forall x \in \mathbb{R}^n.
\]
\end{definition}

Based on the above analysis, we propose a first-order recursive projection algorithm, referred to as the \textit{Gradient Recursive Projection Algorithm for Tampered Binary Observations with Known Probabilities (GRP-TB-KP)}.

\begin{remark}
The loss functions in~\eqref{loss1} and~\eqref{loss2} are not limited to variance-type or squared error forms. In general, any differentiable loss function $L(\beta, y_{k+1}, \varphi_k)$ can be used, provided it satisfies
$
L(\beta, y_{k+1}, \varphi_k) \geq L(\theta, y_{k+1}, \varphi_k),
$
which guarantees that the minimum is achieved at the true parameter $\theta$. Under this condition, stochastic gradient descent can still be employed to obtain consistent estimates. An important example is the negative log-likelihood function.
\end{remark}

\begin{algorithm}[htb]\label{algorithm1}
\caption{GRP-TB-KP}
\begin{algorithmic}[1]
\REQUIRE Initial estimate $\hat{\theta}_1 \in \Theta$, step-size sequence $\{b_k\}_{k \geq 1}$, constant $\beta > 0$
\FOR{$k = 1, 2, \dots$}
    \STATE Compute:
    \begin{align}
    &\tilde{s}_{k+1} = \beta(1 - (p+q)) \times \left((1 - (p+q)) F_k(C - \hat{\theta}_k^T \varphi_k) + q - s_{k+1} \right). \label{algo1}
  \end{align}
    \STATE Update parameter estimate:
    \beq\label{algo2}
    \hat{\theta}_{k+1} = \Pi_{\Theta}\left\{\hat{\theta}_{k} + b_k \varphi_k \tilde{s}_{k+1}\right\}.
    \eeq
\ENDFOR
\end{algorithmic}
\end{algorithm}

\subsubsection{Convergence Properties}

The convergence of the proposed algorithm is established under the following assumptions.

\noindent\textbf{Assumption A1.} The unknown parameter satisfies \(\theta \in \Theta \subseteq \mathbb{R}^n\), where \(\Theta\) is a convex compact set. Let \(B := \sup_{\nu \in \Theta} \|\nu\|\), where \(\|\cdot\|\) denotes the Euclidean norm.

\noindent\textbf{Assumption A2.} The conditional distribution and density functions of the observation noise sequence \(\{w_k\}\), given \(\mathcal{F}_{k}\), are known and denoted by \(F_k(\cdot)\) and \(f_k(\cdot)\), respectively.

\noindent\textbf{Assumption A3.} 
(a) The regressor sequence \(\{\varphi_k\}\) satisfies
\begin{equation} \label{Asm31}
\sup_{k \geq 1} \|\varphi_k\| \leq M < \infty.
\end{equation}
(b) There exist an integer \(h \geq p\) and a constant \(\delta > 0\) such that
\begin{equation} \label{Asm32}
\tfrac{1}{h}\mathbb{E}\left[\left.\sum_{l=k}^{k+h-1} \varphi_l \varphi_l^T \right| \mathcal{F}_{k-1}\right] \geq \delta I, \quad \forall k \geq 1,
\end{equation}
where \(I\) is the \(p \times p\) identity matrix.

\noindent\textbf{Assumption A4.} The step-size sequence \(\{b_k\}\) satisfies:
\[
\sum_{k=1}^{\infty} b_k = \infty, \quad \lim_{k \to \infty} b_k = 0, \quad \text{and} \quad b_k = O(b_{k+1}).
\]

\noindent\textbf{Assumption A5.} The conditional density functions \(\{f_k(\cdot)\}\) of the noise \(\{w_k\}\), given \(\mathcal{F}_{k-1}\), satisfy
\[
\underline{f} := \inf_{k \geq 1} \inf_{|x| \leq C + MB} f_k(x) > 0.
\]

\begin{remark}
Assumption A2 ensures that the statistical properties of the observation noise are known in advance, which is a standard assumption in the literature on binary identification (e.g.,~\cite{Wang2003,Guo2013}). Assumption A3 places conditions on the input sequence \(\{\varphi_k\}\). Condition~\eqref{Asm32}, known as the “conditionally expected sufficiently rich condition,” ensures that the regressor sequence contains enough variability to allow parameter identifiability. Compared to the classical PE condition, it is weaker but still sufficient for convergence.
\end{remark}

\begin{remark}
Assumption A4 is standard in stochastic approximation and online optimization. These conditions ensure that the step size remains effective over time, decays gradually, and avoids premature vanishing. They strike a balance between convergence stability and parameter adaptation. Assumption A5 guarantees that the noise density does not vanish within the relevant domain, preventing degenerate cases. The constant \(\underline{f}\) provides a uniform lower bound, ensuring that the noise remains sufficiently dispersed in the estimation region.
\end{remark}

Denote the estimation error by $\tilde{\theta}_k = \hat{\theta}_k - \theta,\  k = 1, 2, \dots$. The following theorems establish the almost sure convergence, mean-square convergence, and convergence rate of the proposed GRP-TB-KP algorithm.

\begin{theorem}[Convergence]\label{thm1}
Consider system~\eqref{s1} with binary-valued observations~\eqref{s2} and tampering attacks~\eqref{s3}. Under Assumptions A1–A5, the parameter estimate generated by the algorithm~\eqref{algo1}–\eqref{algo2} satisfies
\[
\lim_{k \to \infty} \mathbb{E}[\|\tilde{\theta}_k\|^2] = 0.
\]
Moreover, if $\sum_{k=1}^{\infty} b_k^2 < \infty$, then the estimate $\hat{\theta}_k$ also converges almost surely to the true parameter:
\[
\lim_{k \to \infty} \tilde{\theta}_k = 0, \quad \text{a.s.}
\]
\end{theorem}

\begin{theorem}[Convergence Rate]\label{thm2}
Under Assumptions (A1)-(A5), the algorithm~\eqref{algo1}–\eqref{algo2} achieves the following mean-square convergence rates:
\begin{itemize}
    \item If the step size is chosen as \( b_k = \frac{1}{k^\gamma} \) with \( \frac{1}{2} < \gamma < 1 \), then
    \[
    \mathbb{E}[\|\tilde{\theta}_k\|^2] = O\left( \frac{1}{k^\gamma} \right).
    \]
    
    \item If \( b_k = \frac{1}{k} \) and the gain parameter satisfies \( \beta > \frac{1}{2(1 - p - q)^2 \, \underline{f} \delta} \), then
    \[
    \mathbb{E}[\|\tilde{\theta}_k\|^2] = O\left( \frac{1}{k} \right).
    \]
\end{itemize}
\end{theorem}

\begin{proof}
See Appendix I.
\end{proof}

\subsection{The Attack Strategy is Unknown}



In practice, the tampering probabilities \((p, q)\), which characterize the attack strategy, are typically unknown to the estimator, making existing methods that require their values inapplicable. To address this, we adopt a more practical approach where \((p, q)\) are estimated online along with the system parameters. Following Theorem~\ref{thm1} and the idea in~\cite{guo2020system}, we substitute the real-time estimates of \((p, q)\) into the identification algorithm.

\subsubsection{Algorithm Design}

We now describe the design of the algorithm, starting with its core idea.

\textbf{Main idea.} The approach introduces a known binary sequence (consisting of 0s and 1s) into the transmission channel. By comparing the received sequence with the original one, the estimation center can estimate the tampering probabilities $(p, q)$ using the law of large numbers.

To implement this, we propose a periodic extra-insertion scheme. Known binary signals are inserted at fixed locations, and the corresponding received signals are used to estimate the tampering behavior before proceeding with identification. Specifically, given a period \( T \), define two non-empty, disjoint subsets \( \mathcal{T}_0 \) and \( \mathcal{T}_1 \) of the index set \( \mathcal{T} = \{1, 2, \ldots, T\} \).

During the \( l \)-th period, for each time \( k \in \{(l-1)T + 1, \ldots, lT\} \), the binary output \( s_k^0 \) is first transmitted and possibly tampered as described in~\eqref{s3}. In addition, if \( k - lT \in \mathcal{T}_0 \), a known signal \( s_{k + 1/2}^0 = 0 \) is transmitted; if \( k - lT \in \mathcal{T}_1 \), then \( s_{k + 1/2}^0 = 1 \) is sent. The received signal \( s_{k + 1/2} \) then follows the tampering model:
\begin{equation} \label{attack_p}
\left\{
\begin{aligned}
\Pr\{ s_{k+1/2} = 0 \mid s_{k+1/2}^0 = 1 \} &= p, \\
\Pr\{ s_{k+1/2} = 1 \mid s_{k+1/2}^0 = 0 \} &= q.
\end{aligned}
\right.
\end{equation}

At each time \( k \), define the index sets
\bea
\mathcal{S}_k^0 \ts =\ts \left( \bigcup_{i=1}^{l-1} i\mathcal{T}_0 \right) \bigcup \left( [k - (l-1)T] \cap\mathcal{T}_0 \right),\nn\\
\mathcal{S}_k^1 \ts=\ts \left( \bigcup_{i=1}^{l-1} i\mathcal{T}_1 \right) \bigcup \left( [k - (l-1)T] \cap\mathcal{T}_1 \right).\label{record_set}
\eea

The tampering probabilities can then be estimated as
\bea
\hat{q}_k \ts=\ts \frac{1}{|\mathcal{S}_k^0|} \sum_{i \in \mathcal{S}_k^0} s_{i + 1/2}\nn\\
\ts=\ts\frac{1}{l|\mathcal{T}_0|+|[k - (l-1)T] \cap\mathcal{T}_0|} \sum_{i \in \mathcal{S}_k^0} s_{i + 1/2}, \nn\\
\hat{p}_k \ts=\ts \frac{1}{|\mathcal{S}_k^1|} \sum_{i \in \mathcal{S}_k^1} (1 - s_{i + 1/2})\nn\\
\ts=\ts\frac{1}{l|\mathcal{T}_1|+|[k - (l-1)T] \cap\mathcal{T}_1|} \sum_{i \in \mathcal{S}_k^1} (1-s_{i + 1/2}).\label{estimate_pq}
\eea
These estimates \( \hat{p}_k \) and \( \hat{q}_k \) are then used in the identification algorithm to replace the unknown true values at time \( k \).

Based on the above discussion, we propose the \textit{Gradient Recursive Projection Algorithm for Tampered Binary Observations with Unknown Probabilities (GRP-TB-UP)}.

\begin{remark}
In Algorithm~\ref{algorithm2}, the sets \( \mathcal{T}_0 \) and \( \mathcal{T}_1 \) are designed for online estimation of the attack probabilities \( p \) and \( q \), and are predefined and known to both the system and the estimation center. The sets \( \mathcal{S}^0 \) and \( \mathcal{S}^1 \) correspond to \( \mathcal{S}_k^0 \) and \( \mathcal{S}_k^1 \) as defined in~\eqref{record_set}. The estimation center incrementally records the positions of inserted signals and their corresponding received values, which are then used to estimate the tampering probabilities. The entire algorithm is recursive and operates in an online manner.
\end{remark}

\begin{remark}
In practice, the fixed insertion sets \( \mathcal{T}_0 \) and \( \mathcal{T}_1 \) may become known to the attacker. To enhance robustness, one may adopt varying insertion positions over time, denoted by \( \mathcal{T}_{0l} \) and \( \mathcal{T}_{1l} \) for the \( l \)-th period. As long as the condition
\[
\sum_{j=1}^{l} |\mathcal{T}_{0j}| = O(l), \quad \sum_{j=1}^{l} |\mathcal{T}_{1j}| = O(l)
\]
is satisfied, the theoretical guarantees of the algorithm remain valid.
\end{remark}

\begin{algorithm}[htbp]
\caption{GRP-TB-UP}
\label{algorithm2}
\begin{algorithmic}[1]
\REQUIRE Period \( T \); disjoint subsets \( \mathcal{T}_0, \mathcal{T}_1 \subseteq \mathcal{T} = \{1, 2, \dots, T\} \); initialization: \( \mathcal{S}^0 = \mathcal{S}^1 = \emptyset \), parameter estimate \( \hat{\theta}_1^u \in \Theta \).
\FOR{$l = 1, 2, \dots$}
    \FOR{$k = (l-1)T + 1, \dots, lT$}
        \STATE \textbf{Step 1: Data reception}
        \STATE Transmit the observed data \( s_k^0 \)
        \IF{$k - (l-1)T \in \mathcal{T}_0$}
            \STATE Transmit \( s_{k + 1/2}^0 = 0 \); receive \( s_{k + 1/2} \); update \( \mathcal{S}_0 = \mathcal{S}_0 \cup \{k\} \)
        \ELSIF{$k - (l-1)T \in \mathcal{T}_1$}
            \STATE Transmit \( s_{k + 1/2}^0 = 1 \); receive \( s_{k + 1/2} \); update \( \mathcal{S}_1 = \mathcal{S}_1 \cup \{k\} \)
        \ENDIF
        \STATE \textbf{Step 2: Attack probability estimation}
        \STATE Compute the estimates of \( p \) and \( q \):
        \[
        \hat{q}_k = \frac{1}{|\mathcal{S}^0|} \sum_{i \in \mathcal{S}^0} s_{i + 1/2},\quad
        \hat{p}_k = \frac{1}{|\mathcal{S}^1|} \sum_{i \in \mathcal{S}^1} (1 - s_{i + 1/2})
        \]
        \STATE \textbf{Step 3: Recursive projected identification}
        \STATE Compute:
        \bea
        \ts\ts\tilde{s}_{k+1}^u = \beta(1 - (\hat{p}_k + \hat{q}_k))\label{eq:algo1_mod}\times\left( (1 - (\hat{p}_k + \hat{q}_k)) F_k(C - \hat{\theta}_k^{uT} \varphi_k) + \hat{q}_k - s_{k+1} \right)\nn\\
        \ts\ts\ 
        \eea
        \STATE Update the parameter estimate:
        \begin{equation}
        \hat{\theta}_{k+1}^u = \Pi_{\Theta} \left\{ \hat{\theta}_k^u + b_k \varphi_k \tilde{s}_{k+1}^u \right\}
        \label{eq:algo2_mod}
        \end{equation}
    \ENDFOR
\ENDFOR
\end{algorithmic}
\end{algorithm}

\subsubsection{Convergence analysis}

Denote the estimation error by $\tilde{\theta}_k^u = \hat{\theta}_k^u - \theta,\  k = 1, \dots$. Due to the additional transmission of inserted data, the filtration $\mathcal{F}_k$ in~\eqref{filtration} is modified as
\[
\mathcal{F}_k \triangleq \sigma\left\{ \varphi_i, w_i, s_i \right\}_{i<k+1}, \quad k \geq 0,
\]
which ensures that $s_{k+1/2}$ is $\mathcal{F}_k$-measurable.

The following result establishes the convergence of Algorithm~\ref{algorithm2}.

\begin{theorem}\label{thm3}
Consider the system~\eqref{s1} with binary-valued observation~\eqref{s2}, under the defense scheme~\eqref{attack_p}–\eqref{record_set} and the data tampering model~\eqref{s3}. Let the step size be $b_k = \frac{1}{k^\gamma}$ with $1/2 < \gamma \leq 1$. If Assumptions A1–A3 and A5 hold, then Algorithm~\ref{algorithm2} converges almost surely.
\end{theorem}

\begin{proof}
From~\eqref{attack_p}, we have
\beq\label{prop_pq}
P\{s_{k+1/2} = 1\} =
\begin{cases}
q, & k\mod T \in \mathcal{T}_0, \\
1 - p, & k\mod T \in \mathcal{T}_1.
\end{cases}
\eeq
Thus, by the law of large numbers and~\eqref{record_set}–\eqref{estimate_pq}, it follows that
\beq\label{3thm1}
\hat{q}_k \to q,\quad \hat{p}_k \to p, \quad \text{a.s. as } k \to \infty.
\eeq
In addition, by the law of the iterated logarithm,
\beq\label{3thm2}
\begin{aligned}
|q_k - q| &= O\left( \sqrt{ \frac{\log\log \lfloor \frac{k}{T} \rfloor}{\lfloor \frac{k}{T} \rfloor} } \right) = O\left( \sqrt{ \frac{\log\log k}{k} } \right), \\
|p_k - p| &= O\left( \sqrt{ \frac{\log\log k}{k} } \right).
\end{aligned}
\eeq
Define the error term \( \epsilon_{k+1} = \tilde{s}_{k+1}^u - \tilde{s}_{k+1} \). Then from~\eqref{algo1}, \eqref{eq:algo1_mod}, and the above convergence rates,
\beq\label{3thm3}
\epsilon_{k+1} \to 0, \quad |\epsilon_{k+1}| = O\left( \sqrt{ \frac{\log\log k}{k} } \right).
\eeq
Now consider the update rule in~\eqref{eq:algo2_mod}. Using the non-expansiveness of projection operators, and Proposition~\ref{prop1} in Appendix I, we get
\bea
\|\tilde{\theta}_{k+1}^u\|^2 \ts=\ts \left\| \Pi_{\Theta}\left\{ \hat{\theta}_{k}^u + b_k \varphi_k \tilde{s}_{k+1} + b_k \varphi_k \epsilon_{k+1} \right\} - \theta \right\|^2 \nn \\
\ts\leq\ts \left\| \tilde{\theta}_{k}^u + b_k \varphi_k \tilde{s}_{k+1} + b_k \varphi_k \epsilon_{k+1} \right\|^2 \label{3thm4}\nn \\
\ts=\ts \|\tilde{\theta}_{k}^u\|^2 + 2b_k \tilde{s}_{k+1} \varphi_k^T \tilde{\theta}_{k}^u + O(b_k^2) + O(b_k \epsilon_{k+1}).\nn\\
\ts\ts\ 
\eea
Substituting~\eqref{3thm3} and noting that \( b_k = \frac{1}{k^\gamma} \) with \( \gamma > \frac{1}{2} \), we obtain
\bea
\ts\ts\|\tilde{\theta}_{k+1}^u\|^2\nn\\
\ts\leq\ts \|\tilde{\theta}_{k}^u\|^2 + 2b_k \tilde{s}_{k+1} \varphi_k^T \tilde{\theta}_{k}^u + O\left( \frac{1}{k^{2\gamma}} \right) + O\left( \frac{\sqrt{\log\log k}}{k^{\gamma + 1/2}} \right) \nn \\
\ts=\ts \|\tilde{\theta}_{k}^u\|^2 + 2b_k \tilde{s}_{k+1} \varphi_k^T \tilde{\theta}_{k}^u + O\left( \frac{\sqrt{\log\log k}}{k^{\gamma + 1/2}} \right). \label{3thm5}
\eea
Let \( B_k := \frac{\sqrt{\log\log k}}{k^{\gamma + 1/2}} \). Clearly, \( B_k = O(B_{k+1}) \) and \( b_k^2 = O(B_k) \). From~\eqref{eq:algo1_mod}, we have
\beq\label{3thm6}
\mathbb{E}[\tilde{s}_{k+1} \mid \mathcal{F}_{k}] = \beta(1 - p - q)^2 \left( F(C - \varphi_k^T \hat{\theta}_k^u) - F(C - \varphi_k^T \theta) \right).
\eeq
Proceeding as in Theorem~\ref{thm1}, we derive
\bea
\ts\ts\mathbb{E}[\|\tilde{\theta}_{k+h}^u\|^2] \label{3thm7}\\
\ts=\ts \left(1 - 2\beta(1 - p - q)^2 \underline{f} \delta \sum_{l=k}^{k+h-1} b_l \right)\mathbb{E}[\|\tilde{\theta}_k^u\|^2] + O(B_{k+h}).\nn
\eea
Now observe that $\sum_{k=1}^{\infty} b_k = \infty$ and
\[
\lim_{k \to \infty} \frac{B_k}{\sum_{l=k-h}^{k-1} b_{l+1}} \leq \lim_{k \to \infty} \frac{B_k}{h b_k} = \lim_{k \to \infty} \frac{\sqrt{\log\log k}}{h k^{1/2}} = 0.
\]
Then, by Lemma~\ref{lem2}, we conclude
$
\lim_{k \to \infty} \mathbb{E}[\|\tilde{\theta}_k^u\|^2] = 0.
$
Furthermore, using~\eqref{3thm5} and~\cite[Lemma 1.2.2]{chen2021stochastic}, since
\[
\mathbb{E}[\|\tilde{\theta}_{k+1}^u\|^2 \mid \mathcal{F}_k] \leq \|\tilde{\theta}_k^u\|^2 + O(B_k),\ \sum_{k=1}^\infty B_k < \infty,
\]
we obtain that \( \|\tilde{\theta}_k^u\| \) converges almost surely to a finite limit. Combined with the fact that \( \mathbb{E}[\|\tilde{\theta}_k^u\|^2] \to 0 \), we conclude that \( \tilde{\theta}_k^u \to 0 \) almost surely.
\end{proof}

\section{Second-order Newton Identification Algorithm}

In the previous sections, we proposed gradient-based recursive projection algorithms for both known and unknown attack strategies. These methods adopt scalar step sizes and rely on first-order information, which facilitates implementation. However, as pointed out in~\cite{ljung1983theory}, such first-order methods often suffer from slow convergence, especially in high-dimensional problems or when parameter sensitivities vary significantly. This limitation arises from their inability to utilize curvature information in the cost function.

To address this, we introduce second-order Newton-type algorithms that employ matrix-based gains to adaptively scale the update direction. These algorithms incorporate curvature information and typically achieve faster convergence and improved estimation accuracy.

\subsection{The Attack Strategy is Known}

We first consider the case where the tampering probabilities \((p, q)\) are known. To estimate the unknown parameter \(\theta\) in the binary-valued system~\eqref{s1}–\eqref{s2} under the attack model~\eqref{s3}, we develop a quasi-Newton recursive identification algorithm. A projection step is included to ensure the boundedness of the estimates and to guarantee convergence. Before presenting the algorithm, we introduce the following definitions.

\begin{definition}
For the Euclidean space $\mathbb{R}^p$ ($p \geq 1$), the weighted norm $\|\cdot\|_Q$ associated with a positive definite matrix $Q$ is defined as
\begin{equation}\label{proj_Q}
\|x\|_Q^2 = x^T Q x, \quad \forall x \in \mathbb{R}^p.
\end{equation}
\end{definition}

\begin{definition}
For a convex compact set $\mathcal{D} \subseteq \mathbb{R}^p$ and a positive definite matrix $Q$, the projection operator $\Pi_Q(\cdot)$ is defined as
\begin{equation}
\Pi_Q(x) = \arg\min_{\omega \in \mathcal{D}} \|x - \omega\|_Q, \quad \forall x \in \mathbb{R}^p.
\end{equation}
\end{definition}

\begin{remark}
The projection operator $\Pi_Q(\cdot)$ is non-expansive in the Euclidean norm:
\begin{equation}\label{non_expansive}
\|\Pi_Q(x) - \Pi_Q(y)\| \leq \|x - y\|, \quad \forall x, y \in \mathbb{R}^p.
\end{equation}
\end{remark}

We now present the proposed \textit{Newton-type Recursive Projection Algorithm for Tampered Binary observations with Known Probabilities (NRP-TB-KP)}, as summarized in Algorithm~\ref{algorithm3}.

\begin{algorithm}
\caption{NRP-TB-KP}
\begin{algorithmic}[1]\label{algorithm3}
\STATE Initialize: $\hat{\theta}_1 \in \Theta$, $P_1 > 0$, and $\scriptsize \beta_0\! =\! \sign(1\!-\!(p+q))$ $\min\left\{1, \inf_{|x| \leq LM+C} |1\!-\!(p+q)|f_1(x) \right\}$, 
\FOR{$k = 1, 2, \dots$}
    \STATE $ \beta_{k} = \sign(1-(p+q))\times$\\  $\min\left\{ |\beta_{k-1}|, \inf_{|x| \leq LM+C}|1-(p+q)| f_{k+1}(x) \right\}$
    \STATE $a_k = \frac{1}{1 + \beta_k^2 \varphi_k^T P_k \varphi_k}$
    \STATE $\tilde{s}_{k+1} = (1-(p+q))F_k(C - \hat{\theta}_k^T \varphi_k) + q - s_{k+1}$
    \STATE $P_{k+1} = P_k - \beta_k^2 a_k P_k \varphi_k \varphi_k^T P_k$
    \STATE $\hat{\theta}_{k+1} = \Pi_{P_{k+1}^{-1}}\left( \hat{\theta}_k + a_k \beta_k P_k \varphi_k \tilde{s}_{k+1} \right)$
\ENDFOR
\end{algorithmic}
\end{algorithm}

\begin{remark}
By the matrix inversion formula~\cite{chen2021stochastic}, the inverse of $P_k$ can be updated recursively as
\begin{equation}
P_{k+1}^{-1} = P_k^{-1} + \beta_k^2 \varphi_k \varphi_k^T.
\end{equation}
Since $P_1$ is positive definite, $P_k^{-1}$ remains positive definite for all $k$, which ensures that the projection operator $\Pi_{P_k^{-1}}$ in Algorithm~\ref{algorithm3} is well-defined.
\end{remark}

Since \( \varphi_k \) is \( \mathcal{F}_k \)-measurable, the conditional expectation of \( y_{k+1} \) given \( \mathcal{F}_k \) is
\begin{equation}
    \mathbb{E}(y_{k+1} \mid \mathcal{F}_k) = \theta^T \varphi_k + \mathbb{E}(w_{k+1} \mid \mathcal{F}_k),
\end{equation}
which serves as the optimal predictor in the mean-square sense. Replacing the true parameter with the estimate \( \hat{\theta}_k \), the adaptive predictor becomes
\begin{equation}
    \hat{y}_{k+1} = \hat{\theta}_k^T \varphi_k + \mathbb{E}(w_{k+1} \mid \mathcal{F}_k).
\end{equation}
Letting \( \tilde{\theta}_k = \theta - \hat{\theta}_k \), the instantaneous regret, defined as the squared deviation between the optimal and adaptive predictors, is
\begin{equation}
    R_k = \left(\mathbb{E}(y_{k+1} \mid \mathcal{F}_k) - \hat{y}_{k+1} \right)^2 = \left( \tilde{\theta}_k^T \varphi_k \right)^2.
\end{equation}
A small value of \( R_k \) is desirable and plays a crucial role in evaluating the performance of adaptive control algorithms.

The following three theorems summarize the main theoretical results of this subsection. Without requiring PE condition, we establish asymptotic bounds for the estimation error, cumulative regret, and tracking performance under Algorithm~\ref{algorithm3}.

\begin{theorem}[Estimation Error Bound]\label{thm:estimation}
Under Assumptions A1, A2, A3(a), and A5, the estimation error generated by Algorithm~\ref{algorithm3} satisfies
\begin{equation}
    \|\tilde{\theta}_{n+1}\|^2 = O\left( \frac{\log \lambda_{\max}(P_{n+1}^{-1})}{\lambda_{\min}(P_{n+1}^{-1})} \right), \quad \text{a.s.}
\end{equation}
where \( \tilde{\theta}_k = \theta - \hat{\theta}_k \).
\end{theorem}

\begin{theorem}[Regret Bound]
\label{thm:regret}
Let Assumptions A1, A2, A3(a), and A5 hold. Then the cumulative regret, defined by \( R_k = (\tilde{\theta}_k^T \varphi_k)^2 \), satisfies
\begin{equation*}
    \sum_{k=0}^{n} R_k = O\left( \frac{\log |P_{n+1}^{-1}| }{\beta_n^2} \right), \quad \text{a.s.}
\end{equation*}
\end{theorem}

\begin{remark}
According to the update rule of \( \beta_n \) in Algorithm~\ref{algorithm3}, the regret bound becomes unbounded as \( 1 - (p + q) \to 0 \). This is expected, since when \( p + q \to 1 \), it becomes increasingly difficult to distinguish between tampered and genuine signals, which severely undermines the identifiability of the system.
\end{remark}

\begin{theorem}[Tracking Error in Adaptive Control]\label{thm:tracking}
Under the assumptions of Theorem~\ref{thm:regret}, suppose that the noise densities \( \{f_k(x)\} \) satisfy
\begin{equation}
    \sup_k \mathbb{E}\left[ |w_k|^\alpha \mid \mathcal{F}_{k-1} \right] < \infty, \quad \text{a.s., for some } \alpha > 4,
\end{equation}
and the regressor \( \varphi_k \) is designed such that
\begin{equation}
    \hat{\theta}_k^T \varphi_k + \int_{-\infty}^{\infty} x f_k(x) \, dx = y_{k+1}^*,
\end{equation}
for any bounded reference signal \( \{y_{k+1}^*\} \). Then the average tracking error
\begin{equation}
    J_n = \frac{1}{n} \sum_{k=0}^{n-1} (y_{k+1} - y_{k+1}^*)^2
\end{equation}
satisfies
\begin{equation}
    \left| J_n - \frac{1}{n} \sum_{k=1}^{n} \sigma_k^2 \right| = O\left( \sqrt{\frac{\log \log n}{n}} \right), \quad \text{a.s.},
\end{equation}
where \( \sigma_k^2 = \mathbb{E} \left[ \left( w_k - \mathbb{E}(w_k \mid \mathcal{F}_{k-1}) \right)^2 \mid \mathcal{F}_{k-1} \right] \).
\end{theorem}
\begin{proof}
The proofs of Theorems~\ref{thm:estimation}--\ref{thm:tracking} are provided in Appendix~II.
\end{proof}

\subsection{The Attack Strategy is Unknown}

We now turn to the case where the tampering parameters 
$(p,q)$ are unknown. Following a similar framework to the first-order case, we estimate \( p \) and \( q \) online in parallel with the system identification process. The proposed \textit{Newton-type Recursive Projection Algorithm for Tampered Binary observations with Unknown Probabilities (NRP-TB-UP)} is presented in Algorithm~\ref{algorithm4}. This algorithm offers similar theoretical guarantees to the case with known tampering parameters.

\begin{algorithm}
\caption{NRP-TB-UP}
\begin{algorithmic}[1]\label{algorithm4}
\REQUIRE Period \( T \), disjoint subsets \( \mathcal{T}_0, \mathcal{T}_1 \subseteq \mathcal{T} = \{1, 2, \dots, T\} \), $\mathcal{S}^0 = \mathcal{S}^1 = \emptyset$, initial estimate \( \hat{\theta}_1 \in \Theta \), \( P_1 > 0 \).
\FOR{$l = 1, 2, \dots$}
    \FOR{$k = (l-1)T + 1, \dots, lT$}
        \STATE \textbf{Step 1: Data reception}
        \STATE Transmit the observed data \( s_k^0 \)
        \IF{$k - (l-1)T \in \mathcal{T}_0$}
            \STATE Transmit \( s_{k + 1/2}^0 = 0 \), receive \( s_{k + 1/2} \), and update $\mathcal{S}^0 = \mathcal{S}^0 \cup \{k\}$
        \ELSIF{$k - (l-1)T \in \mathcal{T}_1$}
            \STATE Transmit \( s_{k + 1/2}^0 = 1 \), receive \( s_{k + 1/2} \), and update $\mathcal{S}^1 = \mathcal{S}^1 \cup \{k\}$
        \ENDIF

        \STATE \textbf{Step 2: Online estimation of \( p, q \)}
        \[
        \hat{q}_k = \frac{1}{|\mathcal{S}^0|} \sum_{i \in \mathcal{S}^0} s_{i + 1/2},\quad
        \hat{p}_k = \frac{1}{|\mathcal{S}^1|} \sum_{i \in \mathcal{S}^1} (1 - s_{i + 1/2})
        \]

        \STATE \textbf{Step 3: Recursive parameter update}
        \bea
        \ts\ts\beta_{k} = \sign(1 - (\hat{p}_k + \hat{q}_k)) \cdot \min\left\{ |\beta_{k-1}|, \inf_{|x| \leq LM + C} f_{k+1}(x)(1\! -\! (\hat{p}_k \!+\! \hat{q}_k)) \right\} \nn \\
        \ts\ts a_k = \frac{1}{1 + \beta_k^2 \varphi_k^T P_k \varphi_k} \nn \\
        \ts\ts\tilde{s}_{k+1}^u\! =\! (1 \!-\! (\hat{p}_k + \hat{q}_k)) F_k(C\! -\! \hat{\theta}_k^{uT} \varphi_k)\! +\! \hat{q}_k \!-\! s_{k+1} \nn \\
        \ts\ts P_{k+1} = P_k - \beta_k^2 a_k P_k \varphi_k \varphi_k^T P_k \nn \\
        \ts\ts\hat{\theta}_{k+1}^u = \Pi_{P_{k+1}^{-1}}\left( \hat{\theta}_k^u + a_k \beta_k P_k \varphi_k \tilde{s}_{k+1}^u \right) \nn
        \eea
    \ENDFOR
\ENDFOR
\end{algorithmic}
\end{algorithm}

\begin{theorem}\label{thm4}
Consider system~\eqref{s1} with binary-valued observations~\eqref{s2}, operating under the defense scheme~\eqref{attack_p}–\eqref{record_set} and subjected to the data tampering attack~\eqref{s3}. Suppose that Assumptions~A1, A2, A3(a), and~A5 hold. Then, the estimation error produced by Algorithm~\ref{algorithm4} satisfies
\begin{equation}\label{convergence}
    \|\tilde{\theta}_{n+1}^u\|^2 = O\left( \frac{\log \lambda_{\max}(P_{n+1}^{-1})}{\lambda_{\min}(P_{n+1}^{-1})} \right), \ \text{a.s.,}
\end{equation}
where \( \tilde{\theta}_k^u = \theta - \hat{\theta}_k^u \). Moreover, the cumulative regret, defined as \( R_k = (\tilde{\theta}_k^{uT} \varphi_k)^2 \), admits the bound
\begin{equation}\label{regret_e1}
    \sum_{k=0}^{n} R_k = O\left( \frac{\log |P_{n+1}^{-1}| }{\beta_n^2} \right), \quad \text{a.s.}
\end{equation}
\end{theorem}

\begin{proof}
Define the stochastic Lyapunov function:
\[
V_k^u = \tilde{\theta}_k^{uT} P_k^{-1} \tilde{\theta}_k^u.
\] 
Let \( \epsilon_{k+1}=\tilde{s}_{k+1}-\tilde{s}_{k+1}^u \), where $\tilde{s}_{k+1}$ is given in Algorithm \ref{algorithm3}. Define
\bea
    \varepsilon_{k+1} \ts=\ts (1-(p+q))F_k(C - {\theta}^T \varphi_k) + q - s_{k+1},\label{vareps_2} \\
    \psi_k \ts=\ts (1-(p+q))\nn\\
    \ts\ts\cdot \left(F_{k+1}(C - \hat{\theta}_k^T \varphi_k) - F_{k+1}(C - \theta^T \varphi_k)\right).\label{psi_2}
\eea 
Following the argument in (\ref{5thm1}), we obtain
\bea
    V_{k+1}^u \ts\leq\ts  \tilde{\theta}_k^{uT} P_k^{-1} \tilde{\theta}_k^u {- 2  \beta_k \tilde{\theta}_k^{uT}\varphi_k \psi_k + \beta_k^2 (\tilde{\theta}_k^{uT} \varphi_k)^2 } \nn \\
    \ts\ts +  2 a_k \beta_k^2 (\psi_k+\epsilon_{k+1}) \varphi_k^T P_k \varphi_k \varepsilon_{k+1}   \nn \\
    \ts\ts - 2  \beta_k \varphi_k^T \tilde{\theta}_k^u (\varepsilon_{k+1}-\epsilon_{k+1}) { + a_k \beta_k^2 \varphi_k^T P_k  \varphi_k \varepsilon_{k+1}^2 }\nn\\
    \ts\ts + a_k \beta_k^2 \varphi_k^T P_k \varphi_k\epsilon_{k+1}^2-2a_k \beta_k^2 \varphi_k^T P_k \varphi_k \psi_k\epsilon_{k+1}\nn\\
    \ts\ts +a_k \beta_k^2 \varphi_k^T P_k \varphi_k .\label{6thm1}
\eea

According to~\eqref{3thm2}, there exists \( K \in \mathbb{N}_+ \) such that for all \( k > K \), we have
$
\mathrm{sign}(1 - (p_k + q_k)) = \mathrm{sign}(1 - (p + q)),
$ and $|(1 - (p_k + q_k))|<\frac{4}{3}|1 - (p + q)|$. Based on the definition of  $\beta_k$  in Algorithm \ref{algorithm4}, and using (\ref{psi}) together with the mean value theorem, it follows that
\beq\label{6thm2}
\begin{aligned}
&2 \beta_k \tilde{\theta}_k^{uT} \varphi_k \psi_k = 2 \beta_k (\tilde{\theta}_k^{uT} \varphi_k)^2 f_k(\xi_k)(1 - (p + q)) \\
\geq& \frac{2|p+q-1|}{|(1 - (p_k + q_k))|}\beta_k^2 (\tilde{\theta}_k^{uT} \varphi_k)^2\geq \frac{3}{2}\beta_k^2 (\tilde{\theta}_k^{uT} \varphi_k)^2,
\end{aligned}
\eeq
where \(\xi_k\) lies between \(C - \theta^T \varphi_k\) and \(C - \hat{\theta}_k^T \varphi_k\). 

Furthermore, by (\ref{3thm3}) and the fact that \( |\psi_k| < 1 \), for sufficiently large \( n \), we have
\bea
\ts\ts a_k \beta_k^2 \varphi_k^T P_k \varphi_k\epsilon_{k+1}^2-2a_k \beta_k^2 \varphi_k^T P_k \varphi_k \psi_k\epsilon_{k+1}+a_k \beta_k^2 \varphi_k^T P_k \varphi_k\nn\\
\ts\ts=O(a_k \beta_k^2 \varphi_k^T P_k \varphi_k),\label{6thm3}\\
\ts \ts a_k \beta_k^2 (\psi_k+\epsilon_{k+1}) \varphi_k^T P_k \varphi_k \varepsilon_{k+1}=O( a_k \beta_k^2 \psi_k\varphi_k^T P_k \varphi_k \varepsilon_{k+1}).\nn
\eea

Summing both sides of (\ref{6thm1}) and applying (\ref{6thm2})–(\ref{6thm3}), we obtain
\bea
    \ts\ts V_{n+1}^u \leq  V_0^u- \frac{1}{2}\sum_{k=0}^{n}\beta_k^2 (\tilde{\theta}_k^{uT} \varphi_k)^2+2\sum_{k=0}^{n}  \beta_k \varphi_k^T \tilde{\theta}_k^u \epsilon_{k+1}\nn\\
    \ts\ts +  \underbrace{ 2O\left( \sum_{k=0}^{n} a_k \beta_k^2\psi_k \varphi_k^T P_k \varphi_k \varepsilon_{k+1}\right) - 2\sum_{k=0}^{n}  \beta_k \varphi_k^T \tilde{\theta}_k^u \varepsilon_{k+1} }_{\text{I}} \nn \\
    \ts\ts + \underbrace{ O\left(\sum_{k=0}^{n} a_k \beta_k^2 \varphi_k^T P_k\varphi_k\right) +\sum_{k=0}^{n}  a_k \beta_k^2 \varphi_k^T P_k \varphi_k \varepsilon_{k+1}^2 }_{\text{II}}.\label{6thm}
\eea

From (\ref{vareps_2}) and (\ref{mean1}), we know that
\[
\mathbb{E}(\varepsilon_{k+1} \mid \mathcal{F}_k) = 0,\ \sup_k \, \mathbb{E}\left[ |\omega_{k+1}|^2 \mid \mathcal{F}_k \right] < \infty, \quad \text{a.s.},
\]
which implies that $\{\varepsilon_k, \mathcal{F}_k\}$ forms a martingale difference sequence with finite second moment.

Following the arguments in \cite{guo2020time}, and using Assumption A3(a) \( \sup_{k \geq 1} \|\varphi_k\| \leq M < \infty \) along with (\ref{3thm3}), we have
\bea
\ts\ts 2\sum_{k=0}^{n}  \beta_k \varphi_k^T \tilde{\theta}_k^u \epsilon_{k+1}=o\left(\sum_{k=0}^{n}\beta_k^2 (\tilde{\theta}_k^{uT} \varphi_k)^2\right)+O(1),\nn\\
\ts\ts{\rm I}= o\left(\sum_{k=0}^{n}\beta_k^2 (\tilde{\theta}_k^{uT} \varphi_k)^2\right)+O(1),\nn\\
\ts\ts {\rm II}= O \left( \log |P_{n+1}^{-1}| \right).\nn
\eea

Combining the above estimates yields
\[
    \tilde{\theta}_{n+1}^{uT} P_{n+1}^{-1} \tilde{\theta}_{n+1}^u + \sum_{k=0}^{n} \beta_k^2 (\tilde{\theta}_k^{uT} \varphi_k)^2 = O(\log |P_{n+1}^{-1}|), \quad \text{a.s.}
\]
Finally, since \( \{\beta_k\} \) is non-increasing, the results in (\ref{convergence}) and (\ref{regret_e1}) follow.
\end{proof}

Similar to Theorem~\ref{thm:tracking}, Theorem~\ref{thm4} also implies the following result on the tracking error in adaptive control of binary FIR systems under unknown tampering attacks.

\begin{theorem}\label{thm_tracking1}
Consider system~\eqref{s1} with binary-valued observations~\eqref{s2}, operating under the defense scheme~\eqref{attack_p}–\eqref{record_set} and subjected to the data tampering attack~\eqref{s3}. 
Suppose the conditions of Theorem~\ref{thm:tracking} hold, and that the regressor \( \varphi_k \) is constructed as in Theorem~\ref{thm:tracking} for a bounded reference signal \( \{y_{k+1}^*\} \). Then, the average tracking error satisfies
\[
\left| J_n - \frac{1}{n} \sum_{k=1}^{n} \sigma_k^2 \right| = O\left( \sqrt{\frac{\log \log n}{n}} \right), \quad \text{a.s.}
\]
where \( \sigma_k^2 = \mathbb{E} \left[ \left( w_k - \mathbb{E}(w_k \mid \mathcal{F}_{k-1}) \right)^2 \mid \mathcal{F}_{k-1} \right] \).
\end{theorem}

\section{Numerical simulations and practical case study}

\subsection{Numerical simulations}
This section presents three numerical simulations to validate the convergence of the first-order and second-order algorithms, as well as the tracking error bounds in adaptive control.

\begin{figure}[!tb]
  \centering
  \includegraphics[width=0.8\linewidth]{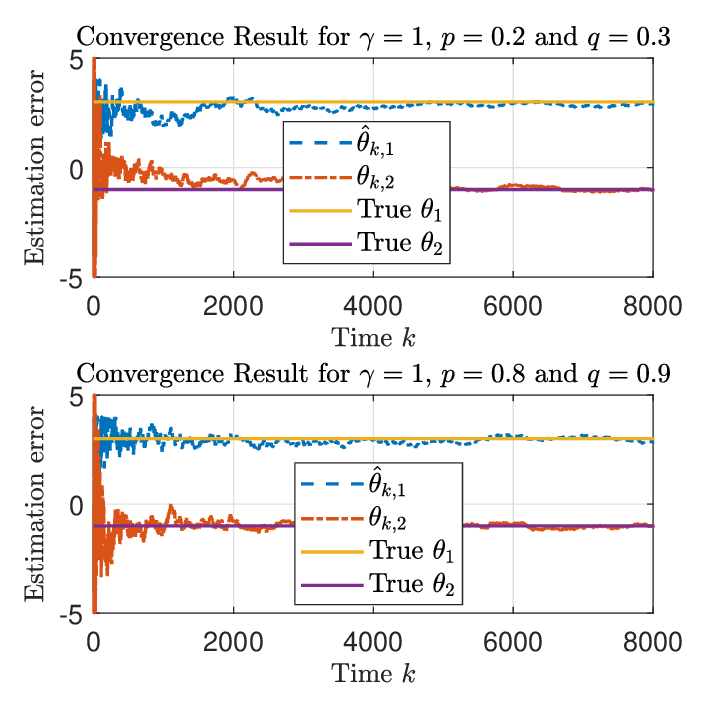}
  \caption{Convergence of the estimation shown by a trajectory of $\hat{\theta}_{n+1}$.}
  \label{fig1a}
\end{figure}

\begin{figure}[!tb]
  \centering
  \includegraphics[width=0.8\linewidth]{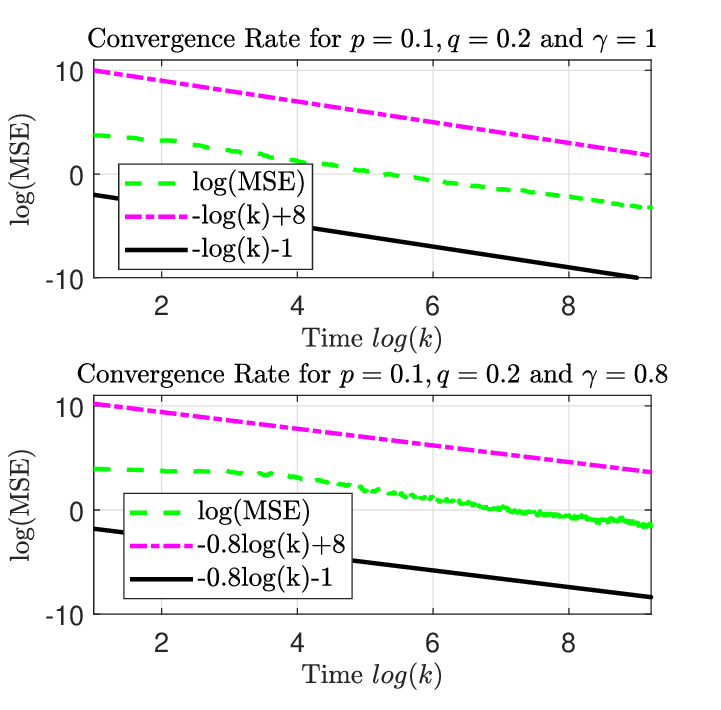}
  \caption{Convergence rate of the estimation shown by a trajectory of $k \tilde{\theta}_k^T \tilde{\theta}_k / \ln k$.}
  \label{fig1b}
\end{figure}

\noindent{\bf Example 1.} Consider the system 
$y_{k+1}=\varphi_{k}^{T} \theta+w_{k+1}$
with the binary observation
$$s_{k}=I_{\left[y_{k} \leq C\right]}= \begin{cases}1, & y_{k} \leq C \\ 0, & \text { otherwise },\end{cases}$$
where $\theta=[3,-1]^T$ is unknown but known as in $\Theta=\{(x,y):|x|<6,|y|<6\}$. The threshold $C=1$, and the system noise $w_{k+1}$ obeys the standard normal distribution. The inputs $\varphi_k=\{u_k,u_{k-1}\}$ with $u_{k}$ obeying the uniform distribution of $N(0,2)$. Algorithm (\ref{algo1})-(\ref{algo2}) has a step size of $\beta=80$, $b_k=1/k^\gamma$ with $\gamma=1,\ 0.8$, and an initial value of $\theta_0=[1,1]^T$. All the simulations are looped 50 times. Figure \ref{fig1a} and \ref{fig1b} present the estimation results of the algorithm for attack strategy shown in (\ref{s3}) as \( p = 0.2, q = 0.3 \) and \( p = 0.8, q = 0.9 \), respectively. From Figure \ref{fig1a}, it can be seen that even when the tampering probability is close to 1, the proposed recursive defense algorithm still converges to the true value (Theorem \ref{thm1}). Figure \ref{fig1b} shows the convergence rate for \( \gamma = 1 \) and \( \gamma = 0.8 \), validating the results derived in Theorem \ref{thm2}. We also consider the case where the attack strategy is unknown, and the attack probabilities $p$ and $q$ are estimated online. The estimator periodically updates $p$ and $q$ using auxiliary data collected at the time indices $\mathcal{T}_0 = \{1,3,5,7,9\}$ and $\mathcal{T}_1 = \{2,4,6,8,10\}$ within each period of length $T = 20$.  In the simulation, two attack scenarios are tested: $(p, q) = (0.2, 0.3)$ and $(p, q) = (0.8, 0.9)$. The simulation results are shown in Figure~\ref{fig:theta_estimation}, where the parameter estimates $\hat{\theta}_k$ converge to the true values, and the online estimates of $p$ and $q$ gradually approach their true values as the sample size increases. These results validate the Theorem \ref{thm3}.

\begin{figure*}[tbp]
    \centering
    \includegraphics[width=0.8\linewidth]{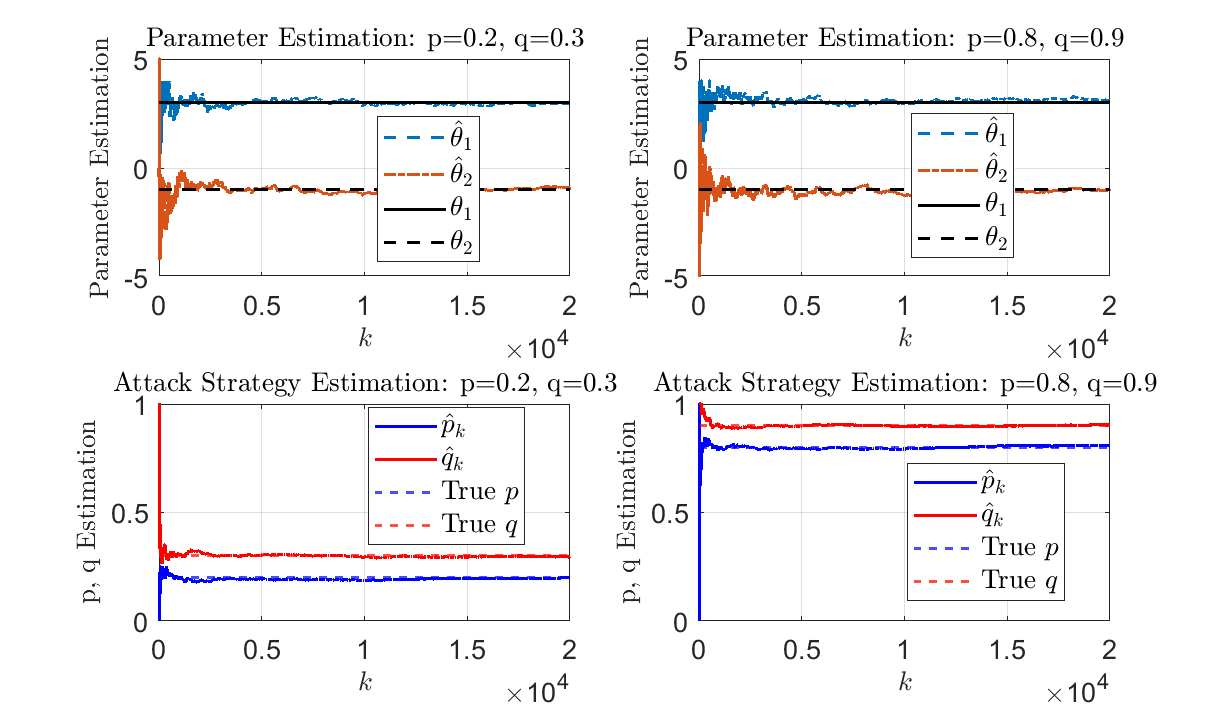}
    \caption{Convergence of parameter estimates $\hat{\theta}_k$ and attack probability estimates $(\hat{p}_k, \hat{q}_k)$ under unknown attack strategies with $(p, q) = (0.2, 0.3)$ and $(p, q) = (0.8, 0.9)$.}
    \label{fig:theta_estimation}
\end{figure*}

\noindent{\bf Example 2.} This example is to validate the theoretical analysis for Algorithm \ref{algorithm4}. We conduct numerical simulations with the same setting as Example 1. The regression vector is defined as $\varphi_k = [u_k, u_{k-1}]^\top$, with $u_k \sim \mathcal{N}(0, \sigma_k^2)$ and $\sigma_k = k^{-1/8}$. At each time step, the binary observation $s_k$ is possibly flipped by an adversary with known probabilities $p=0.1$ and $q=0.2$.
The online estimates of $p$ and $q$ are computed based on partial feedback under a periodic schedule. Figure~\ref{fig:simu_all} presents the simulation results over $N=20000$ time steps. Subplot (a) shows the parameter estimates $\hat{\theta}_1$ and $\hat{\theta}_2$ converging to their true values. Subplot (b) presents the squared estimation error $\|\tilde{\theta}_k\|^2$ along with its theoretical bound $\mathcal{O}(\log k / k{3/4})$, validating Theorem~\ref{thm4}. Subplot (c) depicts the cumulative regret $\sum_{k=0}^n R_k$ and confirms the bound $\mathcal{O}(\log k / \beta_n^2)$ as stated in Theorem~\ref{thm4}. Subplot (d) demonstrates the convergence of the online estimates of the attack probabilities $p$ and $q$, verifying the effectiveness of the periodic extra-insertion scheme.

\begin{figure*}[tbp]
    \centering
    \includegraphics[width=0.8\linewidth]{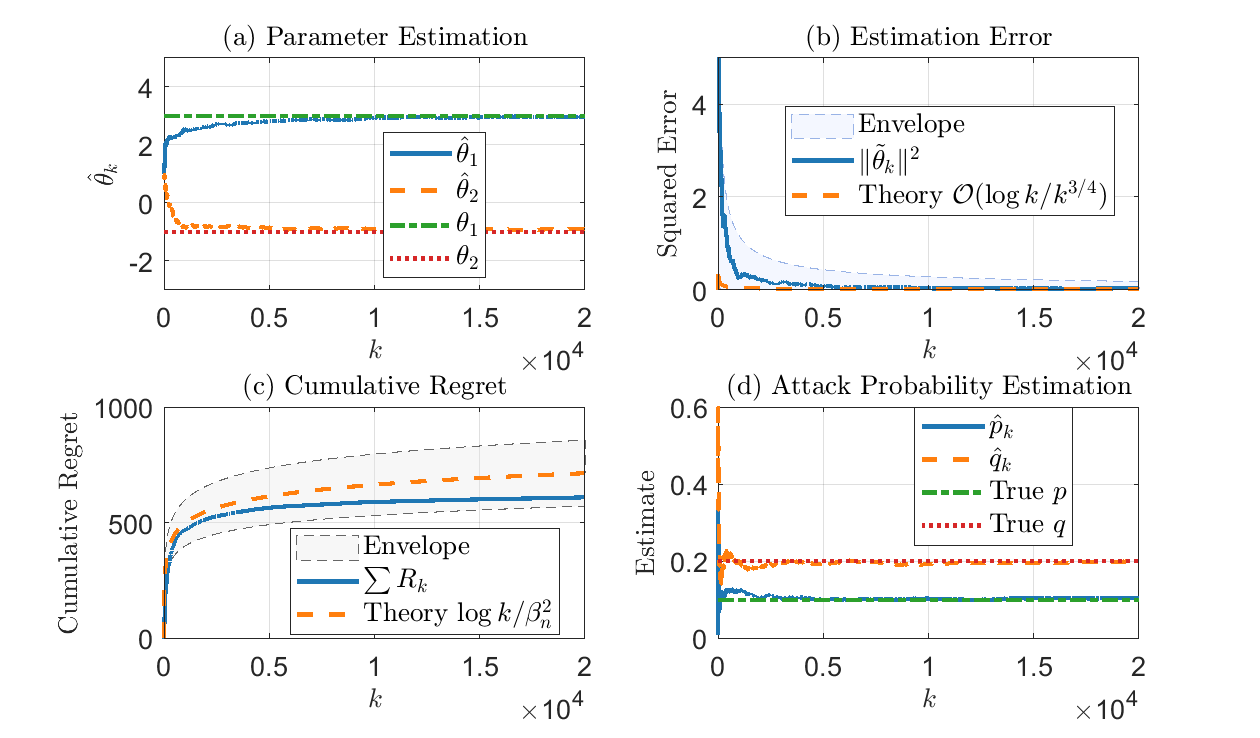}
    \caption{Simulation results of Algorithm \ref{algorithm4}. (a) Parameter estimates; (b) Estimation error and its theoretical bound; (c) Cumulative regret vs theoretical rate; (d) Online estimation of attack probabilities.}
    \label{fig:simu_all}
\end{figure*}

\begin{figure*}[htbp]
  \centering
  \includegraphics[width=0.65\linewidth]{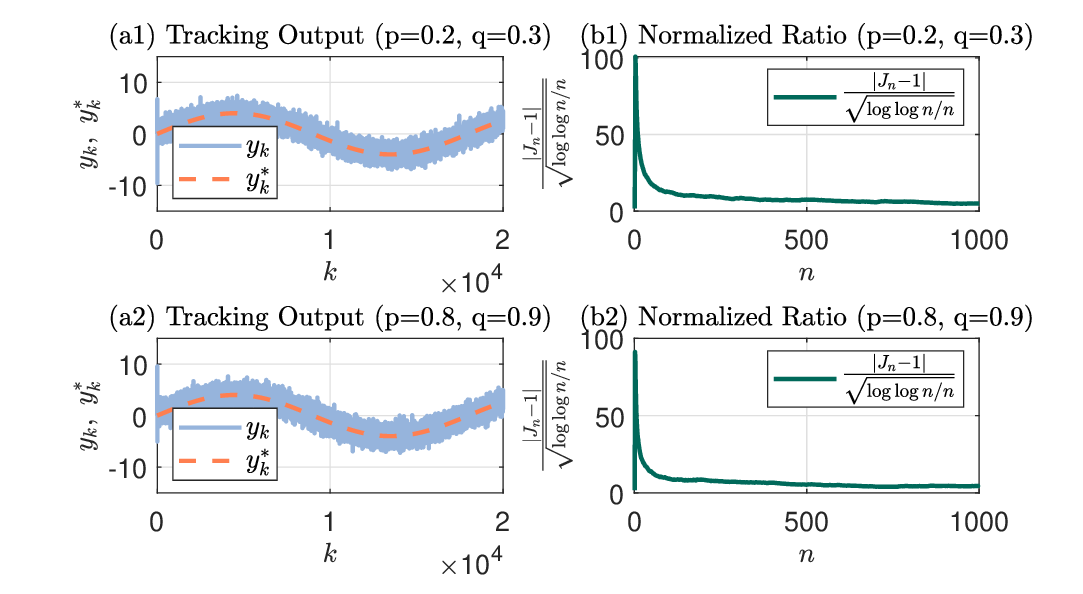}
  \caption{Adaptive tracking with binary feedback under two attack levels. 
  (a1,a2) Output $y_k$ tracks reference $y_k^*$ under $(p,q)=(0.2,0.3)$ and $(0.8,0.9)$. 
  (b1,b2) Normalized error $\frac{|J_n - \sigma^2|}{\sqrt{\log\log n / n}}$.}
  \label{fig:tracking_subplots}
\end{figure*}

\noindent{\bf Example 3.} We also consider the binary observation model with unknown parameter $\theta = [3, -1]^\top$, Gaussian noise $w_k \sim \mathcal{N}(0,1)$, and regressor $\varphi_k = [u_k,\ u_{k-1}]^\top$. We implement the NRP-TB-UP algorithm to estimate $\theta$ online using only binary feedback, while simultaneously designing a control input $u_k$ such that the system output $y_k$ tracks a reference signal $y_k^* = 4\sin(2\pi k/18000)$. The control law solves
\[
\hat{\theta}_k^\top \varphi_k + \int x f_k(x)\, dx = y_{k+1}^*,
\]
where $f_k$ is the density of $w_k$. According to Theorem~\ref{thm_tracking1}, the average tracking error
$
J_n = \frac{1}{n} \sum_{k=1}^{n} (y_k - y_k^*)^2
$
is expected to satisfy
$
|J_n - \sigma^2| = \mathcal{O}\left( \sqrt{ \frac{\log\log n}{n} } \right), \quad \text{a.s.}
$ To verify this, we test two attack settings: $(p,q) = (0.2, 0.3)$ and $(p,q) = (0.8, 0.9)$. Each scenario uses $N=20000$ steps, with initialization $\hat{\theta}_1 = [1,1]^\top$ and random $u_0$. Results are shown in Figure~\ref{fig:tracking_subplots}.

\subsection{Practical Case: Vehicle Emissions Inspection Fraud}

Vehicle emission control, especially for heavy-duty diesel vehicles, plays a critical role in reducing air pollution. However, fraudulent practices in emissions testing—such as using cheat devices to manipulate onboard diagnostic (OBD) systems and falsify excessive emission data—pose serious challenges to environmental monitoring and regulatory enforcement \cite{Zhang2023,Wang2025}. To illustrate the practical relevance of our proposed tampering detection algorithm, we examine a real-world case involving in-use vehicle emissions data. This problem can be naturally formulated within our tampering framework, as falsified emission reports are essentially misclassified compliance labels. We use OBD monitoring data collected in \textit{Hefei, China}, from \textit{August 8 to November 24, 2020}. The dataset includes more than \textit{140,000 records} from over \textit{100 heavy-duty diesel vehicles}, containing both emission measurements and GPS information. Each vehicle was observed for approximately \textit{3 hours}, with data recorded every \textit{5 seconds}, covering a geographic range between \textit{latitudes 31.754772--31.785486} and \textit{longitudes 117.192519--117.245141}.

Given the dynamic driving conditions and the complex temporal structure of emission signals, traditional static models often fail to detect anomalies effectively. To address this, we first process the raw OBD signals using a Transformer-based neural network \cite{vaswani2017attention}, which captures long-range dependencies and nonlinear patterns. This yields a sequence of feature vectors $\varphi_k$, each representing the system state at time $k$.

These features are then used in the model introduced in Section~2.1. We apply the regression model~\eqref{s1} to estimate the predicted emission $y_{k+1}$, followed by a thresholding operation~\eqref{s2} to assign binary compliance labels $s_k^0$. A label of $s_k^0 = 1$ indicates an excessive emission event, i.e., predicted emissions exceeding the regulatory threshold $C$.

However, in real-world inspection scenarios, such binary labels may be deliberately falsified. To model this, we adopt the tampering mechanism~\eqref{s3}, where $p$ denotes the probability of misreporting excessive emissions as compliant, and $q$ the reverse. This completes our pipeline—from raw OBD signals to tampering-aware binary classification—as shown in Figure~\ref{fig:flowchart}.

\begin{figure}[tb]
    \centering
    \includegraphics[width=0.8\textwidth]{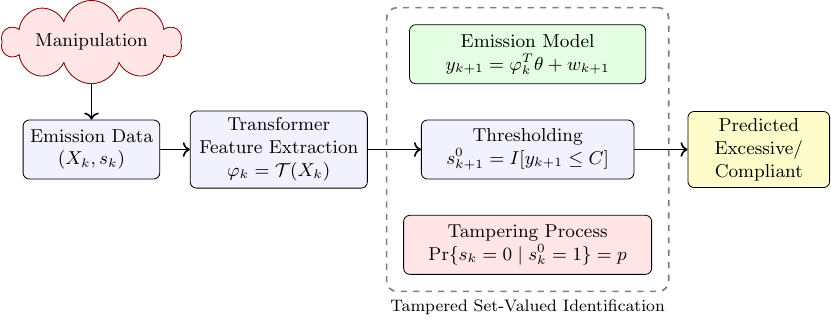}
    \caption{Modeling pipeline for tampered emission data and binary classification.}
    \label{fig:flowchart}
\end{figure}

\begin{figure*}[htbp]
  \centering
  \includegraphics[width=0.9\textwidth]{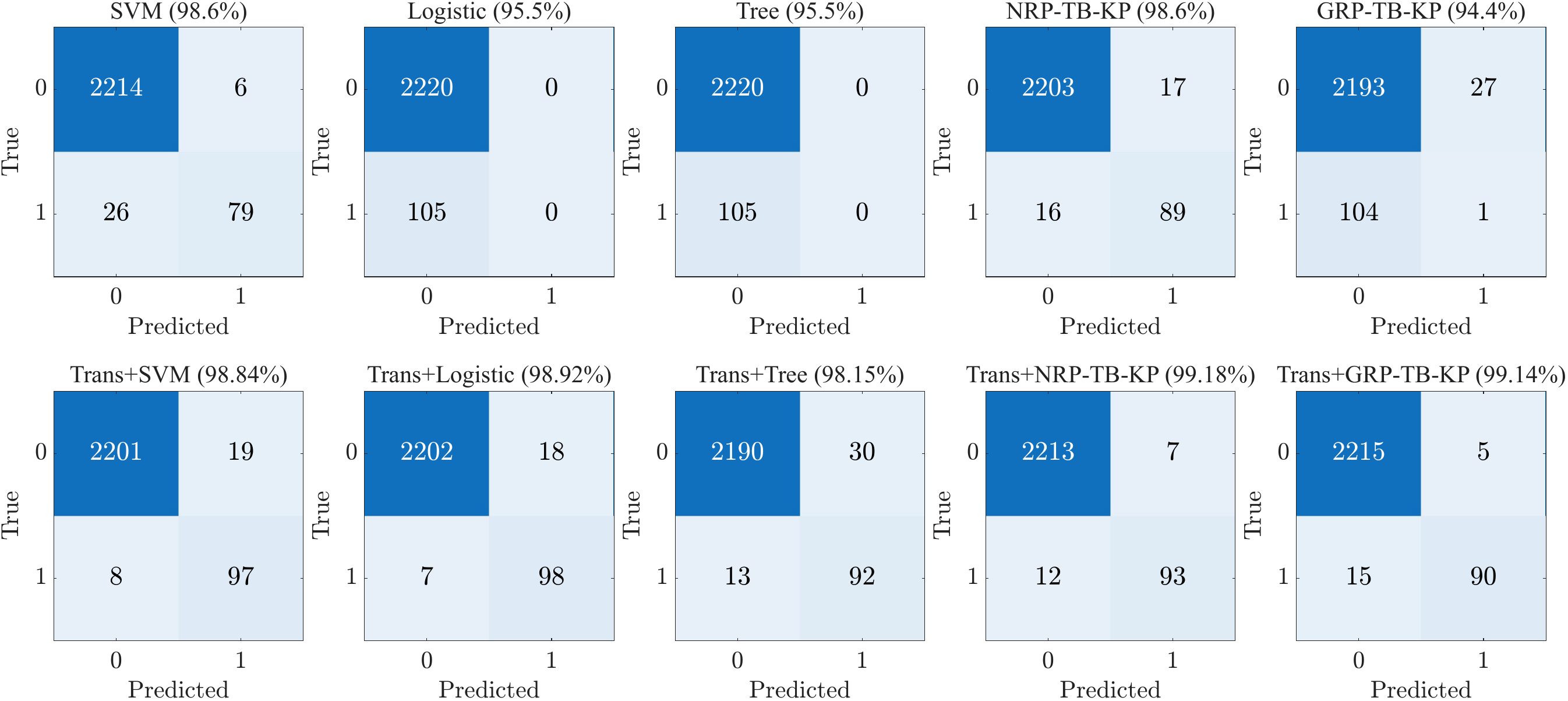}
  \caption{Confusion matrices of five classifiers using Original (top) and Transformer-based features (bottom).}
  \label{fig:confusion_matrices}
\end{figure*}

To assess the effectiveness of the proposed tampering-aware algorithms, we present two sets of simulation results. Figure~\ref{fig:confusion_matrices} shows confusion matrices for five classification algorithms: support vector machine (SVM), logistic regression (Logistic), decision tree (Tree), NRP-TB-KP, and GRP-TB-KP at tampering probabilities $p = 0.3$. The top row results are based on original raw data, while the bottom row represents results after extracting features with a Transformer. Clearly, classification using raw data alone yields relatively low accuracy, emphasizing the need for effective feature extraction methods. After Transformer-based feature extraction, classification accuracy improves across all algorithms. Specifically, the proposed methods based on Transformer-extracted features, NRP-TB-KP and GRP-TB-KP, outperform the baseline classifiers. Among these, the NRP-TB-KP algorithm achieves the highest accuracy overall, highlighting its effectiveness in addressing label tampering and enhancing classification robustness.

\begin{figure*}[tb]
  \centering
  \subfigure[Grad (Orig)]{
    \includegraphics[width=0.19\linewidth]{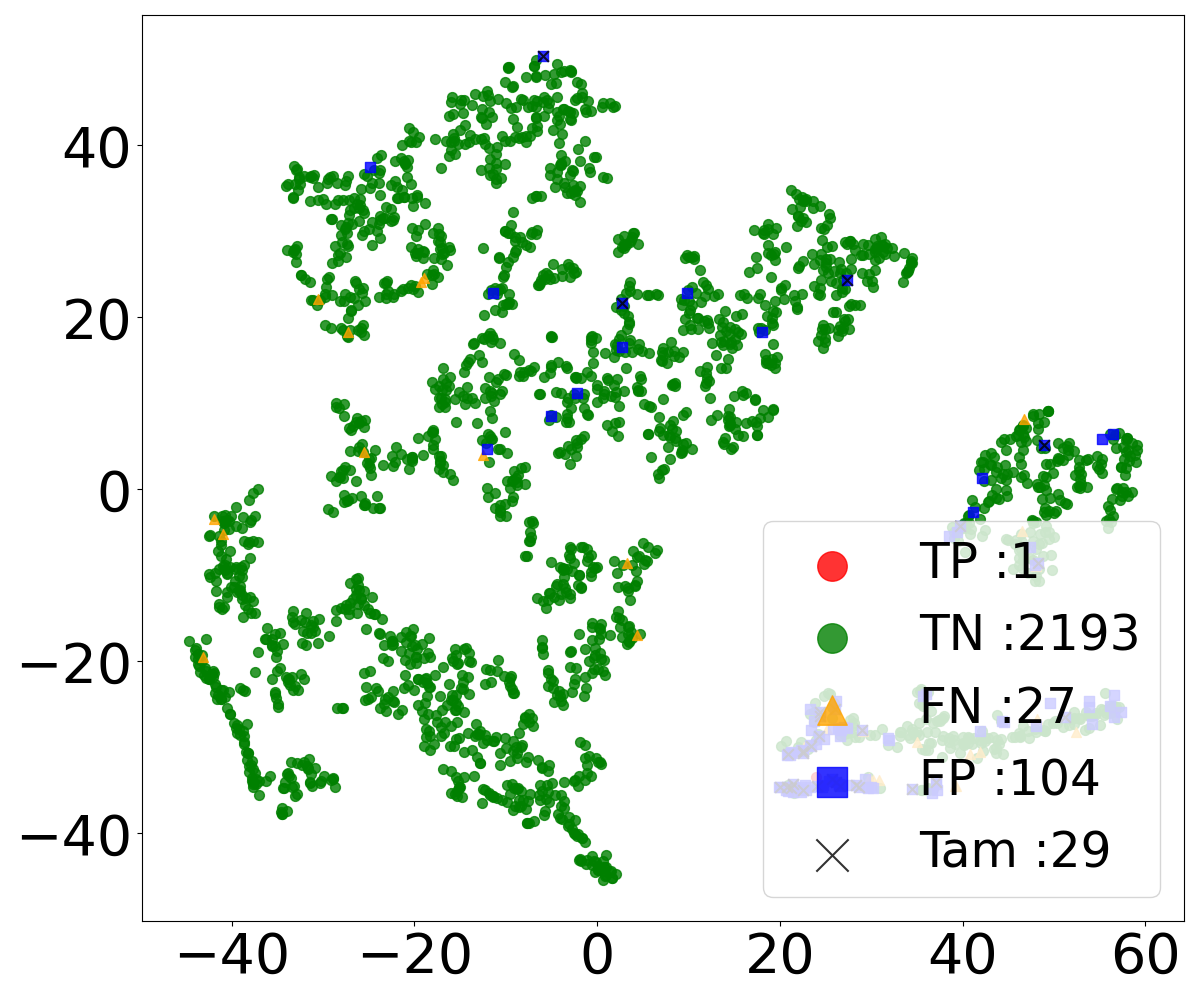}
  }\hspace{-0.6em}
  \subfigure[Logistic (Orig)]{
    \includegraphics[width=0.19\linewidth]{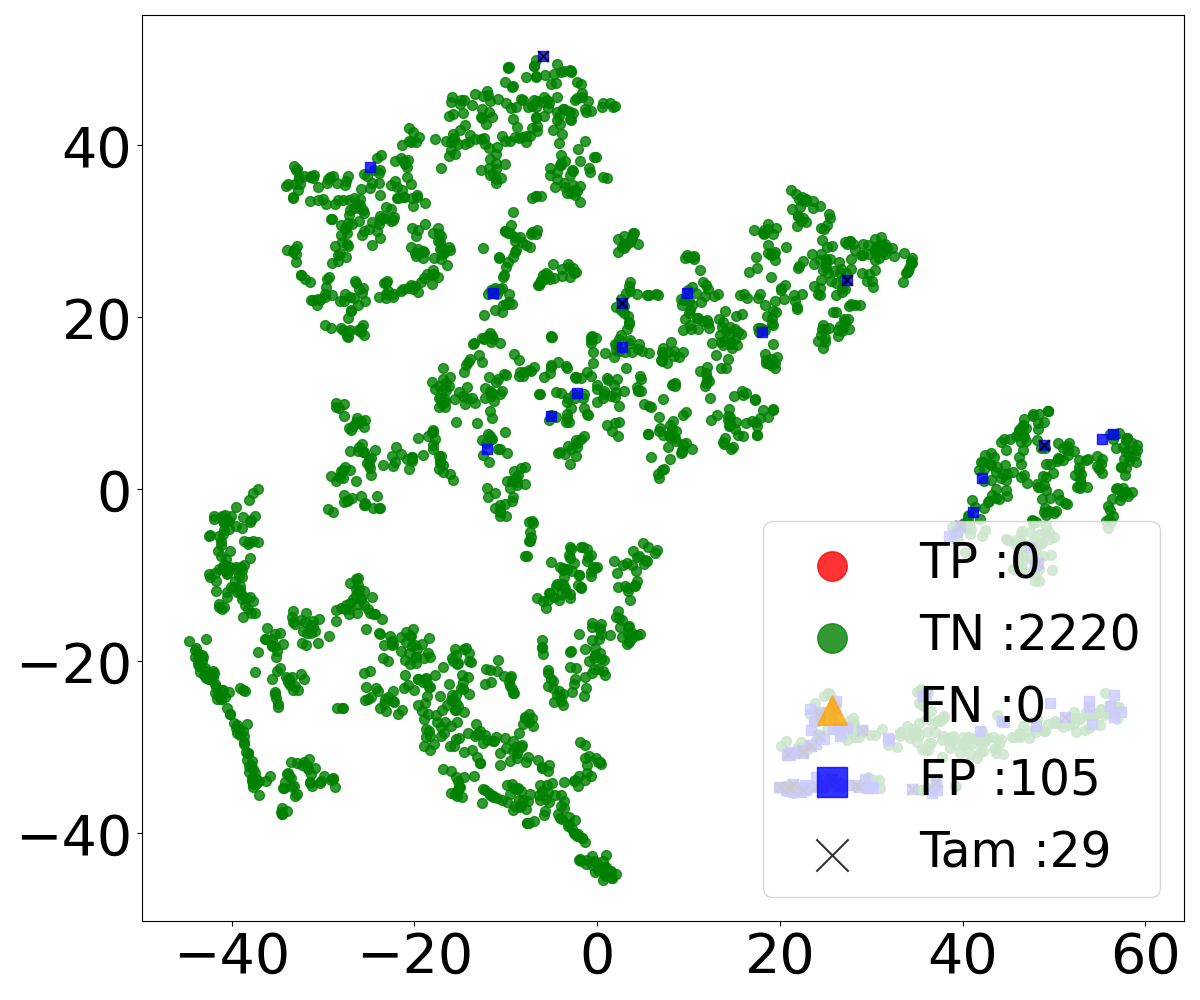}
  }\hspace{-0.6em}
  \subfigure[Newton (Orig)]{
    \includegraphics[width=0.19\linewidth]{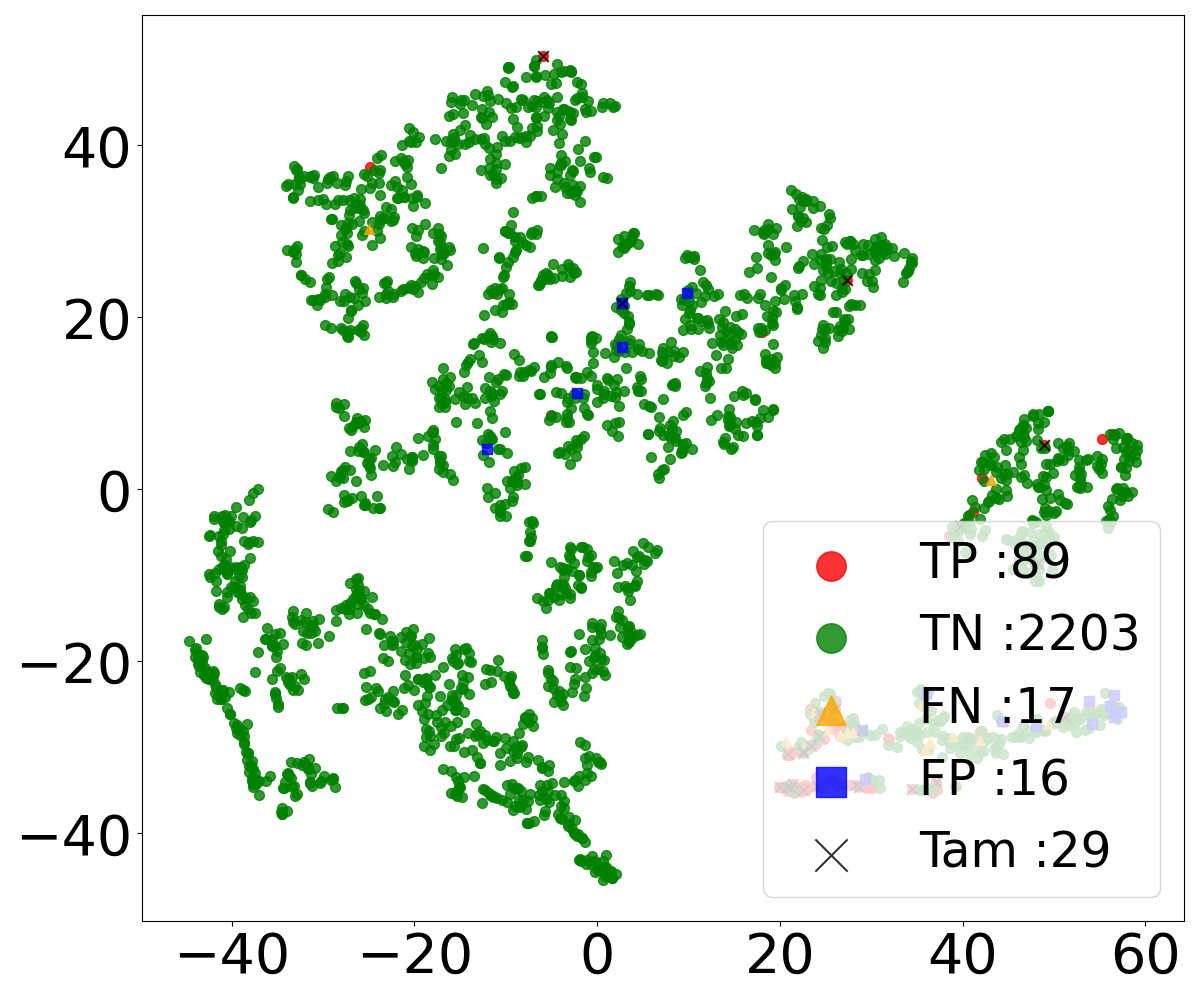}
  }\hspace{-0.6em}
  \subfigure[SVM (Orig)]{
    \includegraphics[width=0.19\linewidth]{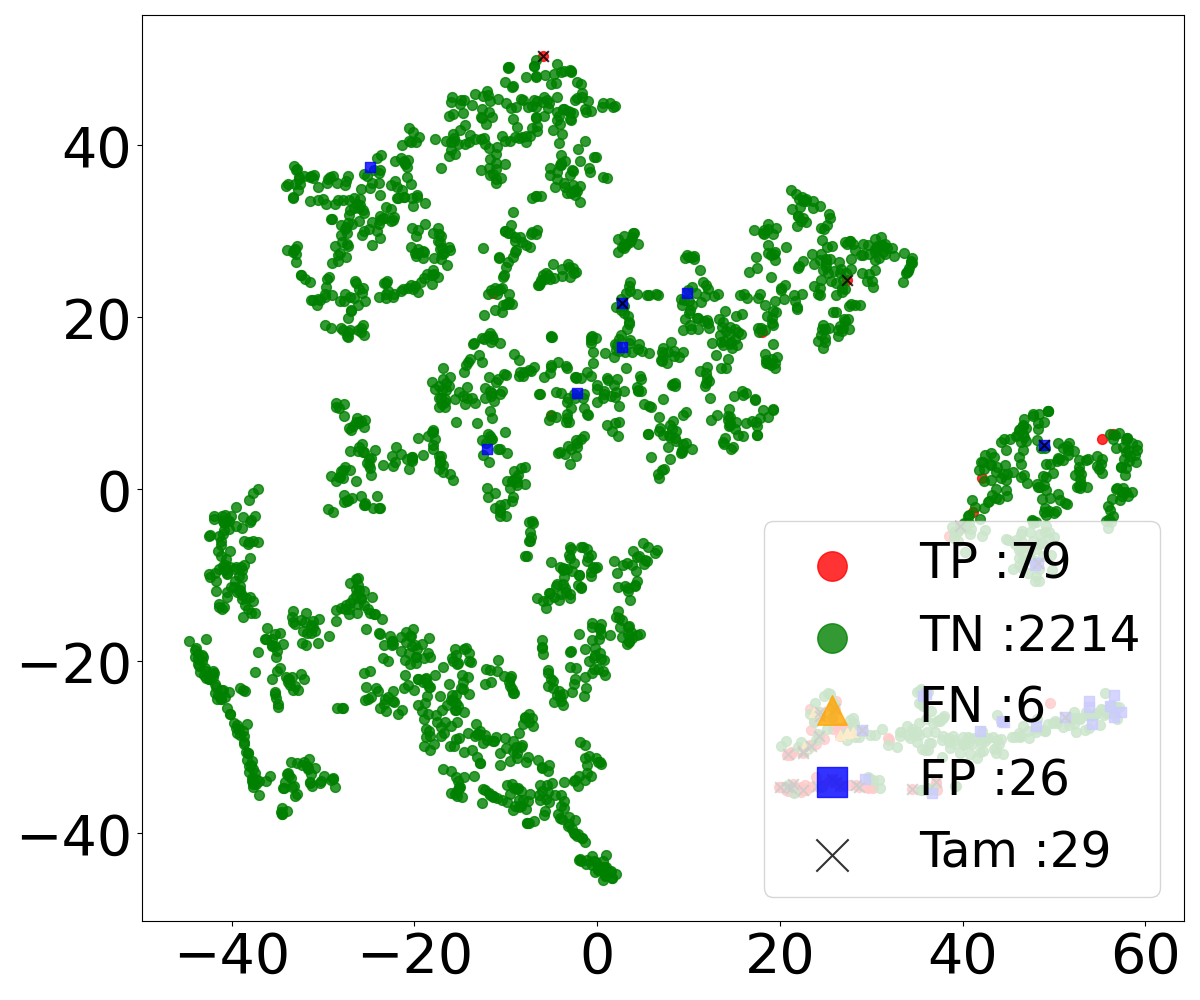}
  }\hspace{-0.6em}
  \subfigure[Tree (Orig)]{
    \includegraphics[width=0.19\linewidth]{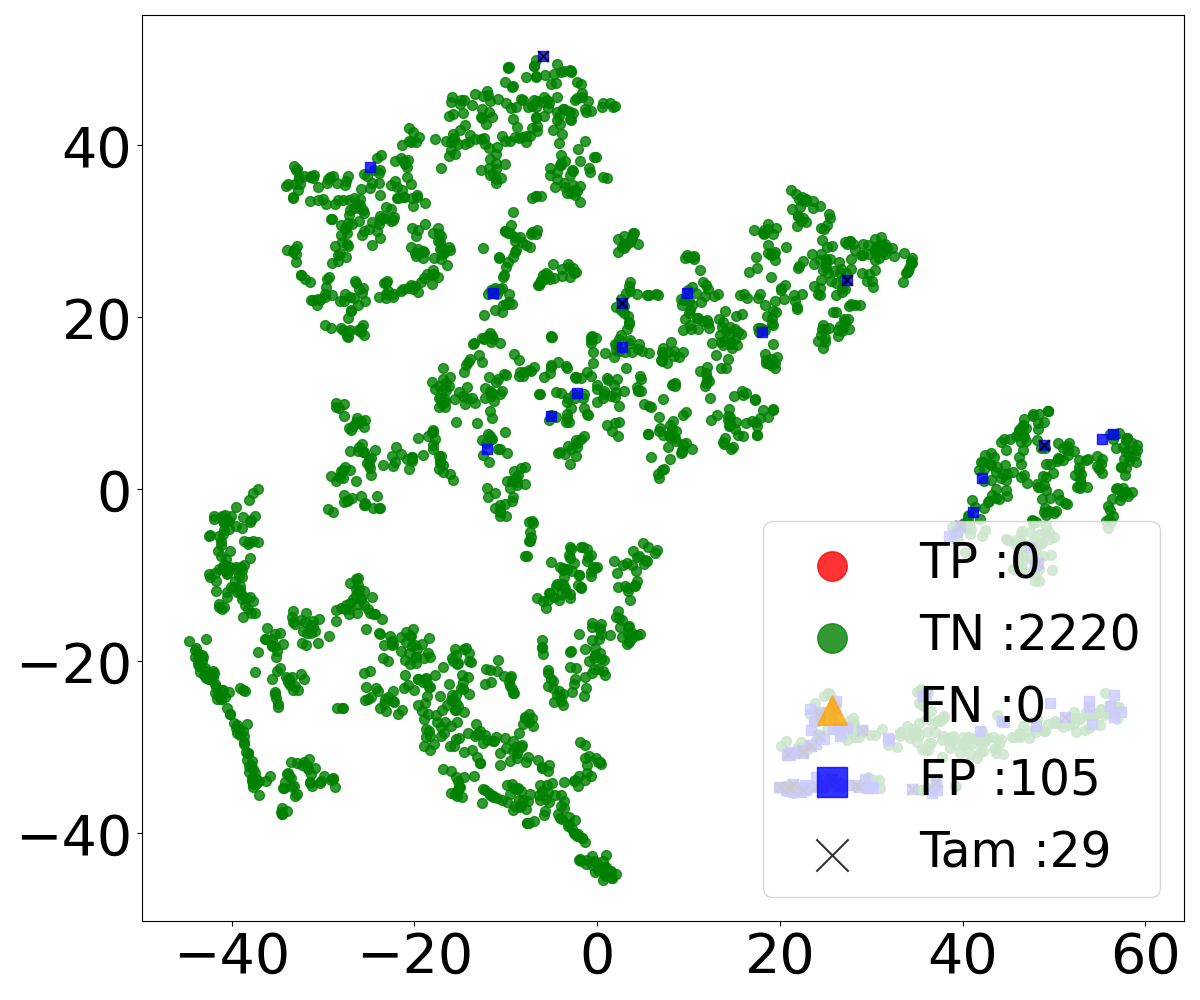}
  }
  \subfigure[Grad (Trans)]{
    \includegraphics[width=0.19\linewidth]{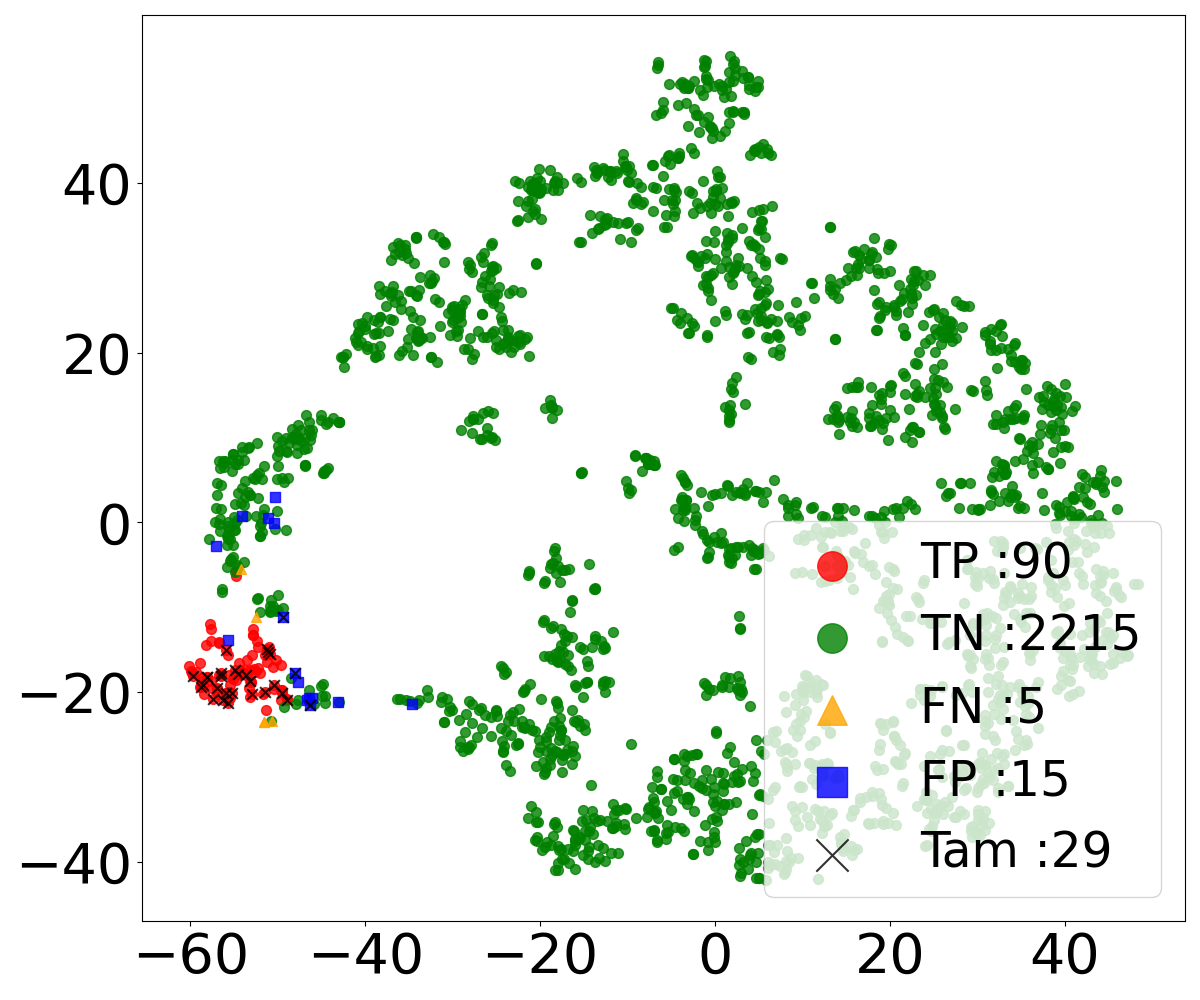}
  }\hspace{-0.6em}
  \subfigure[Logistic (Trans)]{
    \includegraphics[width=0.19\linewidth]{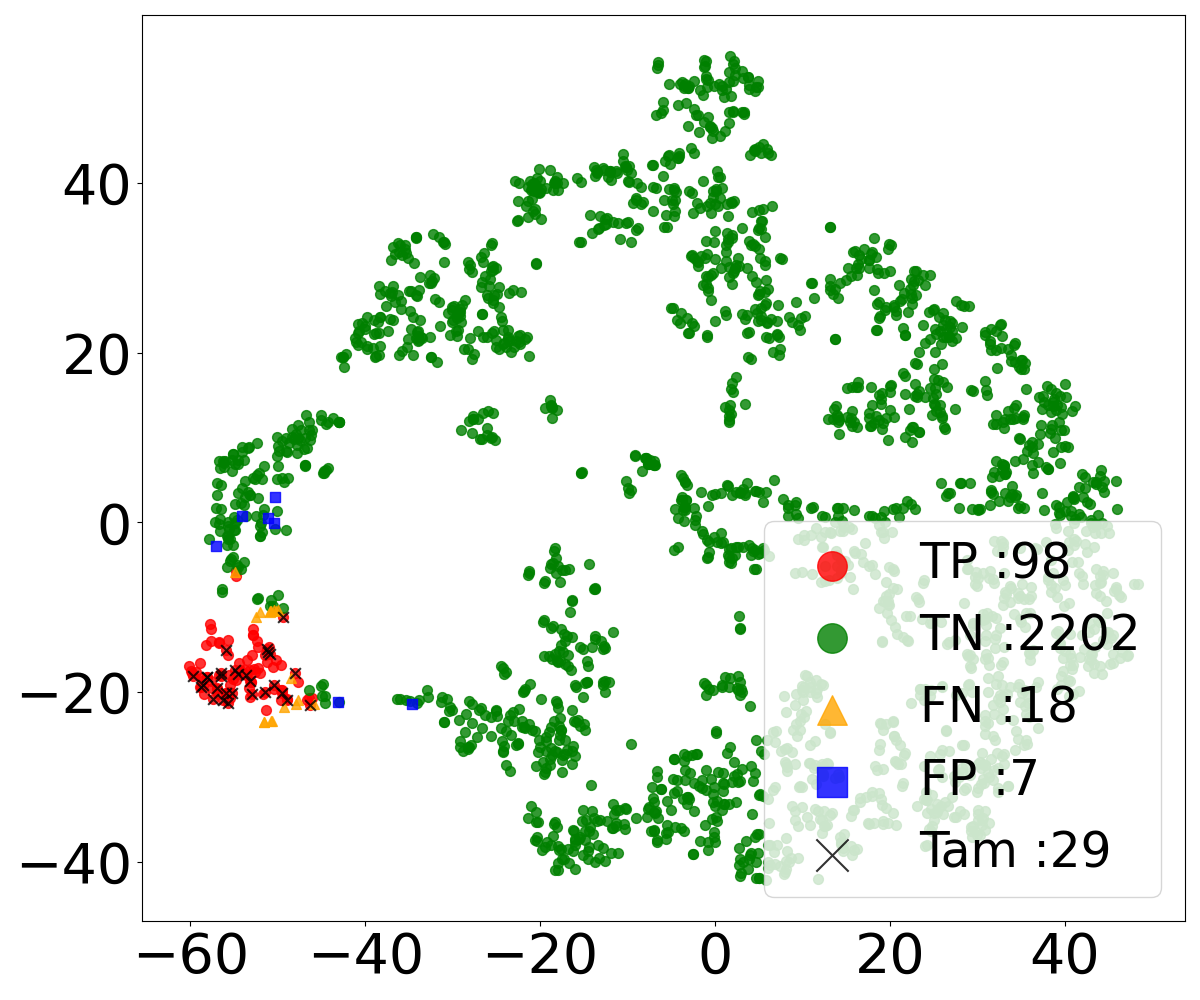}
  }\hspace{-0.6em}
  \subfigure[Newton (Trans)]{
    \includegraphics[width=0.19\linewidth]{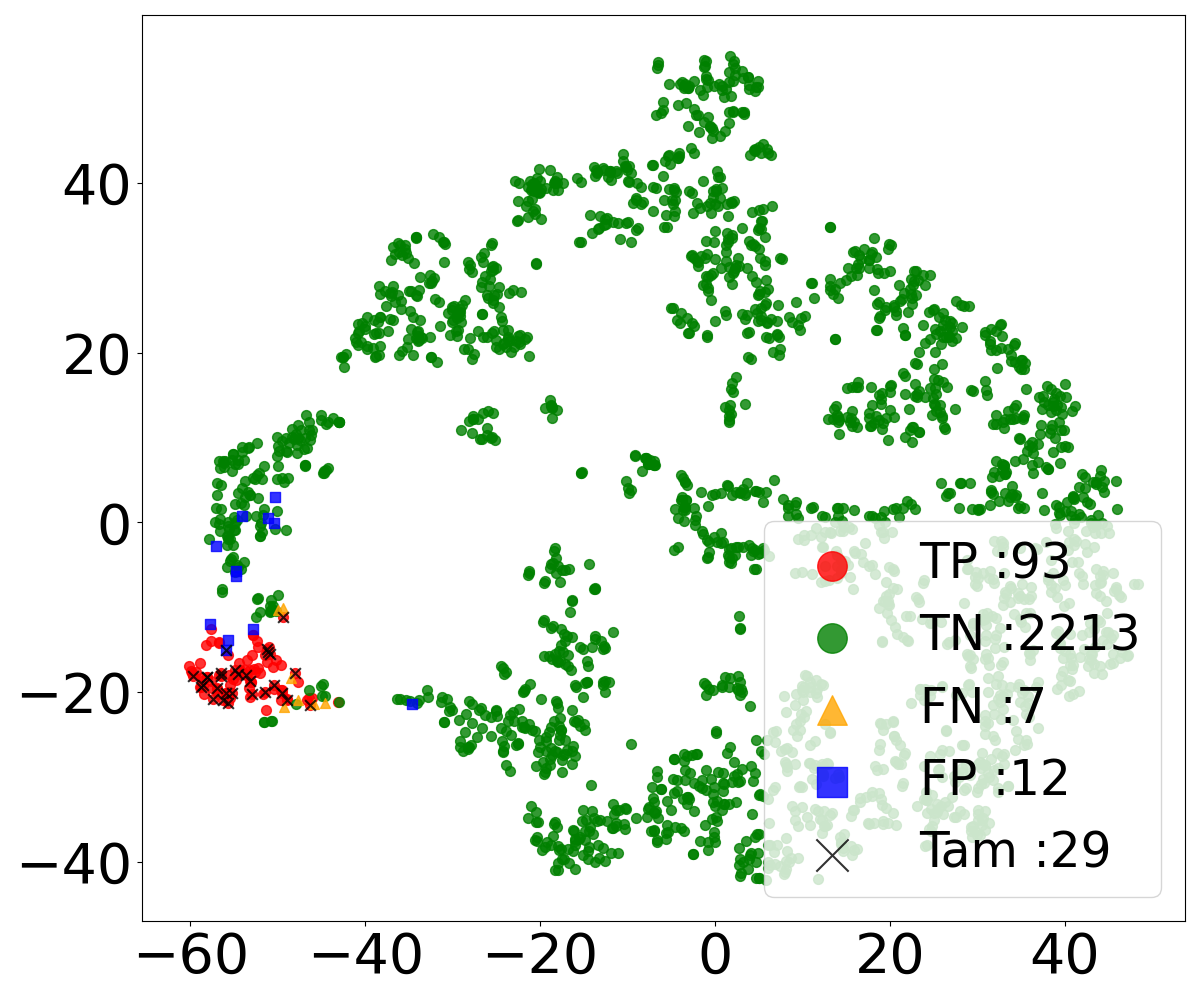}
  }\hspace{-0.6em}
  \subfigure[SVM (Trans)]{
    \includegraphics[width=0.19\linewidth]{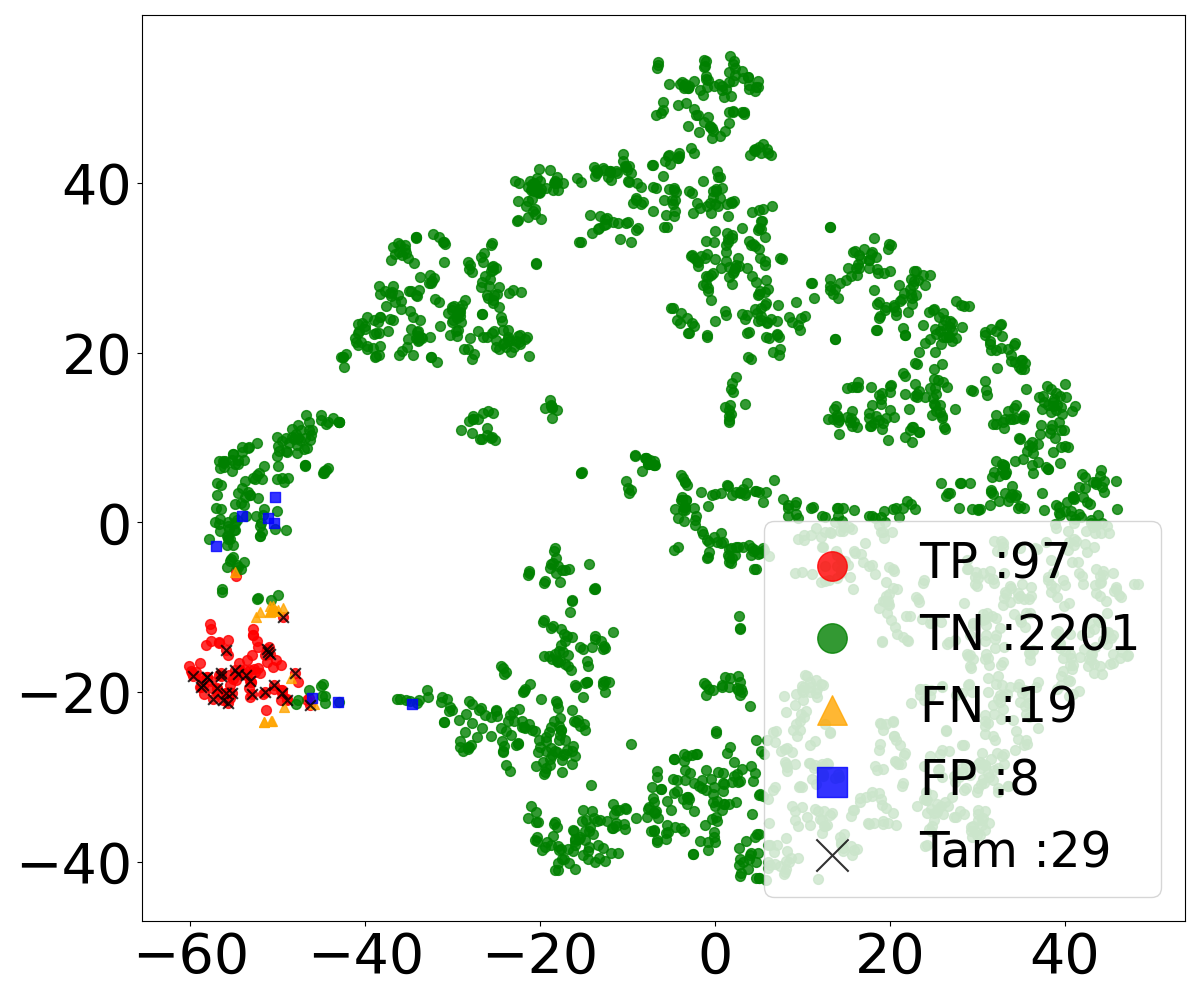}
  }\hspace{-0.6em}
  \subfigure[Tree (Trans)]{
    \includegraphics[width=0.19\linewidth]{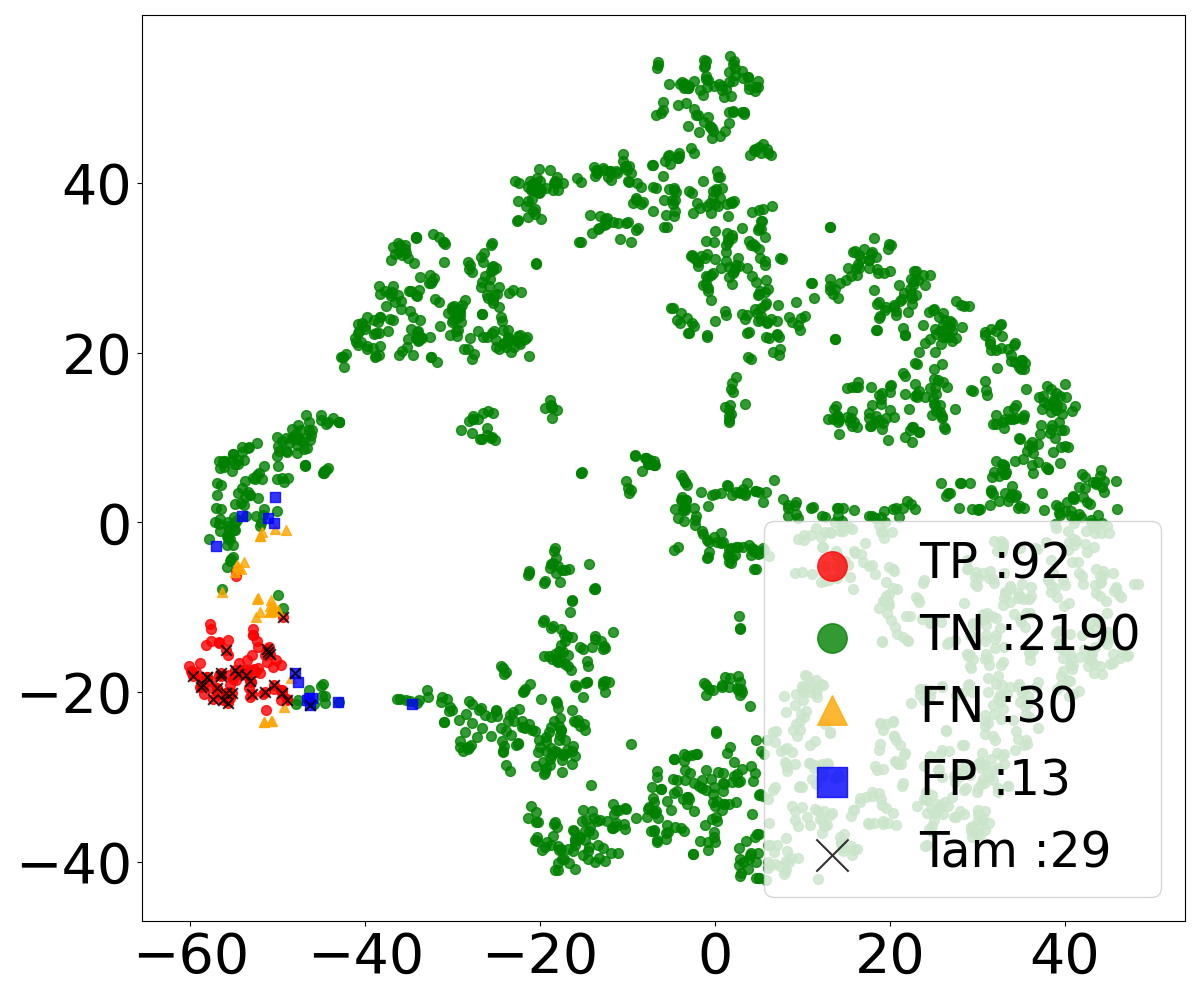}
  }
  \caption{Detection results under different methods. Top: original data; bottom: Transformer-enhanced data. Labels: TP (True Positive), FP (False Positive), TN (True Negative), FN (False Negative), Tam (Tampered).}
  \label{fig:scatter}
\end{figure*}

Figure~\ref{fig:scatter} presents scatter plots corresponding to the confusion matrices in Figure~\ref{fig:confusion_matrices}. Clearly, the proposed NRP-TB-KP algorithm achieves more accurate classification than the compared methods. The points labeled ``Tam'' represent cases with falsified excessive emissions, which are effectively identified and corrected by our proposed approach.

\begin{figure}[!tbp]
    \centering
    \includegraphics[width=0.6\linewidth]{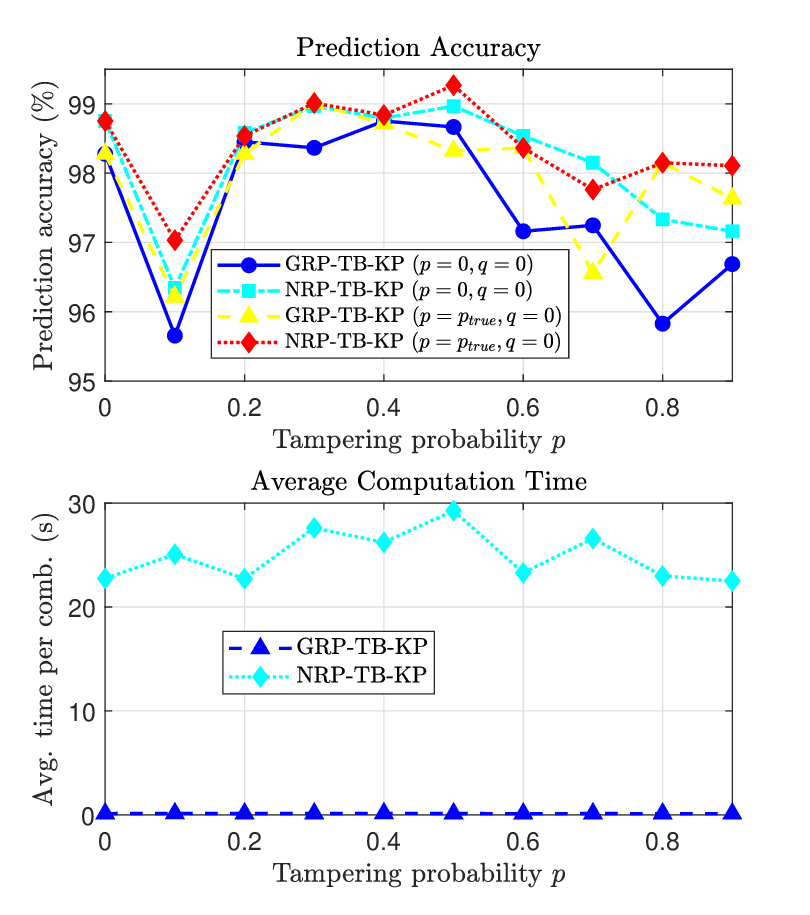}
    \caption{Prediction accuracy and computation time under varying tampering probabilities.}
    \label{fig:tampering_comparison}
\end{figure}

Figure~\ref{fig:tampering_comparison} shows the average accuracy and computation time of different methods under varying tampering probabilities ($p = 0$ to $0.9$). Observations include: (i) Newton-based methods generally achieve higher accuracy, whereas gradient-based methods have shorter computation times; (ii) ignoring tampering consistently leads to decreased performance; and (iii) accuracy typically declines as the tampering probability increases, except the case $p=0.1$. Nonetheless, methods explicitly incorporating tampering probabilities consistently outperform those that do not. These results highlight the importance of considering tampering in emission detection tasks and illustrate the trade-off between accuracy and computational efficiency in practical scenarios.

\section{Concluding remarks}

This paper investigated the problem of parameter estimation and adaptive control for systems with binary observations under data tampering attacks. We developed both gradient-based and second-order Quasi-Newton identification algorithms that are applicable when the attack strategy is either known or unknown. The proposed methods ensure asymptotic convergence of parameter estimates and do not rely on PE condition. In addition, the second-order algorithm was integrated into an adaptive control framework, allowing for explicit tracking error bounds in binary FIR systems even under unknown attacks. Simulation results show the robustness and efficiency of the algorithms, and a vehicle emission control task is used to test their ability.

Future work may extend the framework to multi-agent systems, consider time-varying attack models, or apply it to networked control systems with quantized or event-triggered communication.

\bibliographystyle{IEEEtran}
\bibliography{ifacconf}

\appendix

\section{Proofs of Theorem \ref{thm1} and \ref{thm2} }
Recall the estimate error $\tilde{\theta}_k=\hat{\theta}_k-\theta$. First, we give the following useful lemmas.
\begin{proposition}\label{prop1}
The projection operator given by Definition \ref{def1} follows
$$
\left\|\Pi_{\Omega}(x)-\Pi_{\Omega}(y)\right\| \leq\|x-y\|, \quad \forall x, y \in \mathbb{R}^n.
$$
\end{proposition}
\begin{lemma}\cite{pedregal2004introduction}\label{lem2}
Let \( \{p_k\}, \{q_k\} \) and \( \{\alpha_k\} \) be real sequences satisfying
$
p_{k+1} \leq (1 - q_k) p_k + \alpha_k, \quad \text{where} \quad 0 < q_k \leq 1, \quad \sum_{k=0}^{\infty} q_k = \infty, \quad \alpha_k \geq 0,$ and $\lim_{k \to \infty} \frac{\alpha_k}{q_k} = 0.
$
Then, 
$
\limsup_{k \to \infty} p_k \leq 0.
$
\end{lemma}
\begin{lemma}\cite{wang2024differentially}\label{lem3}
For $0<b \leq 1, a>$ $0, k_0 \geq 0$ and sufficiently large l, we have
$$
\begin{cases}\begin{aligned}
& \prod_{i=l}^k\left(1-\frac{a}{\left(i+k_0\right)^b}\right) \\
&\leq \begin{cases}\left(\frac{l+k_0}{k+k_0}\right)^a, & b=1 , \\
e^{\frac{a}{1-b}\left(\left(l+k_0\right)^{1-b}-\left(k+k_0+1\right)^{1-b}\right)}, & b \in(0,1) ;\end{cases} \\
& \sum_{l=1}^k \prod_{i=l}^k\left(1-\frac{a}{\left(i+k_0\right)^b}\right) O\left(\frac{1}{\left(l+k_0\right)^{2 b}}\right)=O\left(\frac{1}{k^b}\right), \\
& \hspace{15.2em} b \in(0,1) .
\end{aligned}
\end{cases}
$$
\end{lemma}

\begin{lemma}\cite{zhang2019asymptotically}\label{lem4}
For any given positive integer $l$ and $a, b \in \mathbb{R}$, the following results hold
$$
\sum_{l=1}^k \prod_{i=l+1}^k\left(1-\frac{a}{i}\right) \frac{1}{l^{1+b}}=\left\{\begin{array}{l}
O\left(\frac{1}{k^a}\right), a<b \\
O\left(\frac{\ln k}{k^a}\right), a=b \\
O\left(\frac{1}{k^b}\right), a>b.
\end{array}\right.
$$
\end{lemma}

\begin{lemma}\label{lem1}
If Assumptions A1–A3 hold, then
\[
\|\tilde{\theta}_{k+l} - \tilde{\theta}_k\|=O(b_{k+l}),\ \text{for}\ k,\ l\in \mathbb{N}.
\]

\end{lemma}
\begin{proof}
First, we have
\bea
\|\tilde{\theta}_{k+l} - \tilde{\theta}_k\| \ts=\ts \|\hat{\theta}_{k+l} - \hat{\theta}_k\| =\left\| \sum_{j=1}^{l} (\hat{\theta}_{k+j} - \hat{\theta}_{k+j-1}) \right\|\nn\\
\ts\leq\ts\sum_{j=1}^{l}\left\|  (\hat{\theta}_{k+j} - \hat{\theta}_{k+j-1}) \right\|\label{1Lem1}.
\eea
Since $0 \leq (1 - (p+q))F_{k}(C - \theta^T\varphi_k) + q \leq 1$, it follows that $|\tilde{s}_k| \leq \beta$. Combining this with Proposition \ref{prop1} and the condition $\|\phi_k\| \leq M$, we obtain $\|\hat{\theta}_{l+1} - \hat{\theta}_l\| \leq b_{l+1}{\|\phi_{l+1} \tilde{s}_{l+1}\|} \leq b_k\beta M$ for $l \geq 1$. This result, together with (\ref{1Lem1}) and Assumption 4, implies the lemma. 
\end{proof}

\noindent{\it Proof of Theorem \ref{thm1}.} By $\tilde{s}_k^2 \leq \beta^2$, Proposition \ref{prop1} and (\ref{algo1}), we have
\bea
\|\tilde{\theta}_{k+1}\|^2 \ts\leq\ts \|\tilde{\theta}_{k}\|^2 + 2b_k\tilde{s}_{k+1}  \phi_k^T \tilde{\theta}_{k} + b_k^2\|\phi_k\|^2 \beta^2\nn\\
\ts\ts=\|\tilde{\theta}_{k}\|^2 + 2b_k\tilde{s}_{k+1}  \phi_k^T \tilde{\theta}_{k} + O(b_k^2).\label{1thm1}
\eea
From (\ref{mean1}) and (\ref{algo2}), it follows that
\bea\label{mean2}
\ts\ts \mathbb{E}[\tilde{s}_{k+1} | \mathcal{F}_{k}]\nn\\
\ts=\ts \beta(1-p-q)^2 \left(F(C - \phi_k^T \hat{\theta}_{k}) - F(C - \phi_k^T \theta)\right).
\eea

This together with Assumption 2 and the differential mean value theorem, it leads to
\bea
\ts\ts\mathbb{E}\left[2b_k\tilde{s}_{k+1}  \phi_k^T \tilde{\theta}_{k}|\mathcal{F}_k\right]=2b_k\phi_k^T \tilde{\theta}_{k}\mathbb{E}[\tilde{s}_{k+1} | \mathcal{F}_{k}]\nn\\
\ts=\ts 2b_k\phi_k^T \tilde{\theta}_{k}\beta(1-p-q)^2 \left(F_k(C - \phi_k^T \hat{\theta}_{k}) - F_k(C - \phi_k^T \theta)\right)\nn\\
\ts=\ts -2b_k\beta(1-p-q)^2 f_k(\xi_k)\tilde{\theta}_{k}^T\phi_k\phi_k^T\tilde{\theta}_{k} \nn\\
\ts\leq\ts -2b_k\beta(1-p-q)^2 \underline{f}\tilde{\theta}_{k}^T\phi_k\phi_k^T\tilde{\theta}_{k}\label{1thm2},
\eea
where $\xi_k$ is in the interval between $C - \phi_k^T \hat{\theta}_{k}$ and $C - \phi_k^T {\theta}$ such that $F_k(C - \phi_k^T \hat{\theta}_{k}) - F_k(C - \phi_k^T \theta)= f_k(\xi_k)\tilde{\theta}_{k}^T\phi_k\phi_k^T\tilde{\theta}_{k}$.
Taking the expectation on both sides of (\ref{1thm1}) and substituting  (\ref{1thm2}) into it, we can obtain
\bea
\ts\ts\mathbb{E}\|\tilde{\theta}_{k+1}\|^2\leq \mathbb{E}\|\tilde{\theta}_{k}\|^2 + 2b_k\mathbb{E}\tilde{s}_{k+1}  \phi_k^T \tilde{\theta}_{k} + O(b_k^2)\nn\\
\ts=\ts\mathbb{E}\|\tilde{\theta}_{k}\|^2 + \mathbb{E}\left[\mathbb{E}\left[2b_k\tilde{s}_{k+1}  \phi_k^T \tilde{\theta}_{k}|\mathcal{F}_k\right]\right] + O(b_k^2).\nn\\
\ts\leq\ts \mathbb{E}\|\tilde{\theta}_{k}\|^2  -2b_k\beta(1-p-q)^2 \underline{f}\mathbb{E}[\tilde{\theta}_{k}^T\phi_k\phi_k^T\tilde{\theta}_{k}]\nn\\
\ts\ts + O(b_k^2)\label{1thm3}.
\eea
By iterating (\ref{1thm3}) \(h\) times and noting \( b_k = O(b_{k+1}) \), we obtain
\bea
\ts\ts\mathbb{E}\|\tilde{\theta}_{k+h}\|^2\nn\\
\ts\leq\ts \mathbb{E}\|\tilde{\theta}_{k}\|^2  -2\beta(1-p-q)^2 \underline{f}\mathbb{E}\left[\sum\limits_{l=k}^{k+h-1}[b_l\tilde{\theta}_{l}^T\phi_l\phi_l^T\tilde{\theta}_{l}]\right]\nn\\
\ts\ts + O(b_{k+h}^2)\nn\\
\ts\leq\ts \mathbb{E}\|\tilde{\theta}_{k}\|^2  -2\beta(1-p-q)^2 \underline{f}\mathbb{E}\left[\sum\limits_{l=k}^{k+h-1}[b_l\tilde{\theta}_{k}^T\phi_l\phi_l^T\tilde{\theta}_{k}]\right]\nn\\
\ts\ts -2\beta(1-p-q)^2 \underline{f}\mathbb{E}\left[\sum\limits_{l=k}^{k+h-1}[b_l(\tilde{\theta}_{l}-\tilde{\theta}_{k})^T\phi_l\phi_l^T(\tilde{\theta}_{l}-\tilde{\theta}_{k})]\right]\nn\\
\ts\ts + O(b_{k+h}^2).\label{1thm4}
\eea
By Lemma \ref{lem1} and (\ref{Asm31}), the last two terms of (\ref{1thm3}) are of order \( O(b_{k+h}^2) \). In addition, 
\bea
\ts\ts \mathbb{E}\left[\sum\limits_{l=k}^{k+h-1}[b_l\tilde{\theta}_{k}^T\phi_l\phi_l^T\tilde{\theta}_{k}]\right]\nn\\
\ts=\ts \mathbb{E}\left[\tilde{\theta}_{k}^T\mathbb{E}\left[\left.\sum\limits_{l=k}^{k+h-1}b_l\phi_l\phi_l^T\right|\mathcal{F}_k\right]\tilde{\theta}_{k}\right].\label{1thm5}
\eea
By Assumption A3, we have
\bea
\ts\ts\mathbb{E}\left[\left.\sum\limits_{l=k}^{k+h-1}b_l\phi_l\phi_l^T\right|\mathcal{F}_k\right]\nn\\
\ts=\ts\sum\limits_{l=k}^{k+h-1}b_l\frac{1}{h}\mathbb{E}\left[\left.\sum\limits_{l=k}^{k+h-1}\phi_l\phi_l^T\right|\mathcal{F}_k\right]+O(b_{k+h}^2)\nn\\
\ts\geq\ts \delta \sum\limits_{l=k}^{k+h-1}b_lI+O(b_{k+h}^2).\label{1thm6}
\eea
Substituting (\ref{1thm5}) and (\ref{1thm6}) into  (\ref{1thm4}) yields
\bea
\ts\ts\mathbb{E}\|\tilde{\theta}_{k+h}\|^2 \label{1thm7}\\
\ts\leq\ts \mathbb{E}\|\tilde{\theta}_{k}\|^2  -2\beta(1-p-q)^2 \underline{f}\delta\sum\limits_{l=k}^{k+h-1}b_l\mathbb{E}\|\tilde{\theta}_{k}\|^2
 + O(b_{k+h}^2)\nn\\
 \ts=\ts \left(1-2\beta(1-p-q)^2 \underline{f}\delta\sum\limits_{l=k}^{k+h-1}b_l\right)\mathbb{E}\|\tilde{\theta}_{k}\|^2
 + O(b_{k+h}^2)\nn.
\eea

Then, based on Lemma \ref{lem2} and Assumption A4, and noting
$
\sum_{k=1}^{\infty} b_k = \infty$ and $\lim_{k \to \infty} \frac{b_k^2}{\sum_{l=k-h}^{k-1} b_l + 1} = 0,
$ it follows that
$
 \lim_{k \to \infty} \mathbb{E}[\|\tilde{\theta}_k\|^2] = 0.
$

On the other hand, by (\ref{1thm3}) we have
$
\mathbb{E}[\|\tilde{\theta}_{k+1}\|^2 | \mathcal{F}_{k}] \leq \|\tilde{\theta}_{k}\|^2 + O(b_k^2),
$
which together with \cite[Lemma 1.2.2]{chen2021stochastic} and $\sum\limits_{k=1}^{\infty}b_k^2<\infty$ implies that $\|\tilde{\theta}_k\|$ converges to a bounded limit a.s. Notice that $\lim_{k \to \infty} \mathbb{E}[\tilde{\theta}_k^T \tilde{\theta}_k] = 0$. Then, there is a subsequence of $\tilde{\theta}_k$ that converges almost surely to 0. Consequently, $\tilde{\theta}_k$ almost surely converges to 0.\hfill$\square$

\noindent{\it Proof of Theorem \ref{thm2}.}
When \( b_k = \frac{1}{k^\gamma} \), \( \gamma \in (1/2, 1) \), letting $\alpha=2\beta(1-p-q)^2 \underline{f}\delta$ and based on (\ref{1thm7}), we have
\bea
\ts\ts \mathbb{E}\|\tilde{\theta}_k\|^2 \leq \left( 1 - \alpha \sum_{l=k-h}^{k-1} \frac{1}{(l + 1)^\gamma}\right) \mathbb{E}\|\tilde{\theta}_k\|^2 + O\left( \frac{1}{k^{2\gamma}} \right),\nn\\
\ts\leq\ts \left( 1 -  \frac{\alpha h}{k^\gamma}\right)\mathbb{E}\|\tilde{\theta}_{k-h}\|^2 + O\left( \frac{1}{k^{2\gamma}} \right)\nn\\
\ts\leq\ts \prod_{l=0}^{\left\lfloor \frac{k-K}{h} \right\rfloor - 1} \left( 1 - \frac{\alpha h}{(k - lh)^\gamma} \right) \mathbb{E}\left\|\tilde{\theta}_{k-\left\lfloor \frac{k-K}{h} \right\rfloor h}\right\|^2 \nn\\
\ts\ts +\sum_{l=1}^{\left\lfloor\frac{k-K}{h}\right\rfloor} \prod_{j=0}^{l-1}\left(1-\frac{\alpha h}{(k-j h)^\gamma}\right) O\left(\frac{1}{(k-l h)^{2 \gamma}}\right) \nn\\
\ts \leq\ts \prod_{l=\left\lceil\frac{K}{h}\right\rceil+\kappa+1}^{\left\lceil\frac{k}{h}\right\rceil}\left(1-\frac{\alpha h}{(l h)^\gamma}\right)\mathbb{E}\left\|\tilde{\theta}_{k-\left\lfloor \frac{k-K}{h} \right\rfloor h}\right\|^2 \nn\\
\ts\ts +\sum_{l=\left\lceil\frac{K}{h}\right\rceil+1}^{\left\lfloor\frac{k}{h}\right\rfloor-1} \prod_{j=\left\lceil\frac{K}{h}\right\rceil+\kappa+l+1}^{\left\lceil\frac{k}{h}\right\rceil}\left(1-\frac{\alpha h}{(j h)^\gamma}\right) O\left(\frac{1}{(l h)^{2 \gamma}}\right) \nn\\
\ts \leq\ts \prod_{l=\left\lceil\frac{K}{h}\right\rceil+\kappa+1}^{\left\lceil\frac{k}{h}\right\rceil}\left(1-\frac{\alpha h^{1-\gamma}}{l^\gamma}\right)\mathbb{E}\left\|\tilde{\theta}_{k-\left\lfloor \frac{k-K}{h} \right\rfloor h}\right\|^2\nn \\
\ts\ts +\sum_{l=\left\lceil\frac{K}{h}\right\rceil+1}^{\left\lfloor\frac{k}{h}\right\rfloor-1} \prod_{q=\left\lceil\frac{K}{h}\right\rceil+\kappa+l+1}^{\left\lceil\frac{k}{h}\right\rceil}\left(1-\frac{\alpha h^{1-\gamma}}{j^\gamma}\right) O\left(\frac{1}{l^{2 \gamma}}\right),\nn
\eea
where $\kappa=\left\lceil\frac{k-K}{h}\right\rceil-\left\lfloor\frac{k-K}{h}\right\rfloor$. This together with Lemma \ref{lem3} yields $\mathbb{E}\|\tilde{\theta}_k\|^2=O\left(\frac{1}{k^\gamma}\right)$.

When $b_k=\frac{1}{k}$, letting $\alpha=2\beta(1-p-q)^2 \underline{f}\delta$ and by (\ref{1thm7}),
$$
\begin{aligned}
&\mathbb{E}\|\tilde{\theta}_k\|^2 \leq  \left(1-\alpha \sum_{l=k-h}^{k-1} \frac{1}{l+1}\right)\mathbb{E}\|\tilde{\theta}_{k-h}\|^2+O\left(\frac{1}{k^2}\right),\\
\leq & \left(1-\frac{\alpha h}{k}\right)\mathbb{E}\|\tilde{\theta}_{k-h}\|^2+O\left(\frac{1}{k^2}\right) \\
\leq & \prod_{l=0}^{\left\lfloor\frac{k-K}{h}\right\rfloor-1}\left(1-\frac{\alpha h}{k-l h}\right)\mathbb{E}\|\tilde{\theta}_{k-\left\lfloor\frac{k-K}{h}\right\rfloor h \|}\|^2\\
& +\sum_{l=1}^{\left\lfloor\frac{k-K}{h}\right\rfloor} \prod_{q=0}^{l-1}\left(1-\frac{\alpha h}{k-q h}\right) O\left(\frac{1}{(k-l h)^2}\right)\\
\leq & \prod_{l=\left\lceil\frac{k}{h}\right\rceil+\kappa+1}^{\left\lceil\frac{k}{h}\right\rceil}\left(1-\frac{\alpha h}{l h}\right)\mathbb{E}\|\tilde{\theta}_{k-\left\lfloor\frac{k-K}{h}\right\rfloor h \|}\|^2 \\
& +\sum_{l=\left\lceil\frac{K}{h}\right\rceil+1}^{\left\lfloor\frac{k}{h}\right\rfloor-1} \prod_{q=\left\lceil\frac{K}{h}\right\rceil+\kappa+l+1}^{\left\lceil\frac{k}{h}\right\rceil}\left(1-\frac{\alpha h}{q h}\right) O\left(\frac{1}{(l h)^2}\right) 
\end{aligned}
$$
$$
\begin{aligned}
\leq & \prod_{l=}^{\left\lceil\frac{K}{h}\right\rceil+\kappa+1}\left(1-\frac{\alpha}{l}\right)\mathbb{E}\|\tilde{\theta}_{k-\left\lfloor\frac{k-K}{h}\right\rfloor h \|}\|^2 \\
& +\sum_{l=\left\lceil\frac{K}{h}\right\rceil+1}^{\left\lfloor\frac{k}{h}\right\rfloor-1} \sum_{q=\left\lceil\frac{K}{h}\right\rceil+\kappa+l+1}^{\left\lceil\frac{k}{h}\right\rceil}\left(1-\frac{\alpha}{q}\right) O\left(\frac{1}{l^2}\right),
\end{aligned}
$$
where $\kappa=\left\lceil\frac{k-K}{h}\right\rceil-\left\lfloor\frac{k-K}{h}\right\rfloor$. Since $\beta>\frac{1}{2(1-p-q)^2\underline{f}\delta}$, i.e. $\alpha>1$. Thus, by Lemma \ref{lem4}, we have  $\mathbb{E}\left\|\tilde{\theta}_k\right\|^2=O\left(\frac{1}{k}\right)$. This completes this part's proof. \hfill$\square$

\section{Proofs of Theorem  \ref{thm:estimation}-\ref{thm:tracking} }
For convenience, we introduce the notation
\bea
    \varepsilon_{k+1} \ts=\ts (1-(p+q))F_k(C - {\theta}^T \varphi_k) + q - s_{k+1},\label{vareps} \\
    \psi_k \ts=\ts (1-(p+q))\nn\\
    \ts\ts\left(F_{k+1}(C - \hat{\theta}_k^T \varphi_k) - F_{k+1}(C - \theta^T \varphi_k)\right).\label{psi}
\eea
We then give the following lemma.
\begin{lemma}\label{lem}
Let Assumptions A1, A2, A3(a), and A5 be satisfied. Then the parameter estimate given by Algorithm \ref{algorithm3} has the following property as \( n \to \infty \):
\begin{equation}\label{error}
    \tilde{\theta}_{n+1}^T P_{n+1}^{-1} \tilde{\theta}_{n+1} + {\beta_n^2} \sum_{k=0}^{n} \left( \tilde{\theta}_k^T \varphi_k \right)^2 = O \left( \log |P_{n+1}^{-1}| \right), \quad \text{a.s.}
\end{equation}
\end{lemma}

\begin{proof}
Recall \( \tilde{\theta}_k = \theta - \hat{\theta}_k \) and define the stochastic Lyapunov function:
\[
V_k = \tilde{\theta}_k^T P_k^{-1} \tilde{\theta}_k.
\] By (\ref{non_expansive}) and noting $a_k P_{k+1}^{-1} P_k \varphi_k =  \varphi_k$, $P^{-1}_{k+1} = P^{-1}_k + \beta_k^2 \varphi_k \varphi_k^T$ and $|\psi_k|\leq 1$,
we get
\bea
    V_{k+1} \ts\leq\ts  \tilde{\theta}_k^T P_k^{-1} \tilde{\theta}_k {- 2  \beta_k \tilde{\theta}_k^T\varphi_k \psi_k + \beta_k^2 (\tilde{\theta}_k^T \varphi_k)^2 } \nn \\
    \ts\ts + { 2 a_k \beta_k^2 \psi_k \varphi_k^T P_k \varphi_k \varepsilon_{k+1} - 2  \beta_k \tilde{\theta}_k^{uT}\varphi_k  \varepsilon_{k+1} } \nn \\
    \ts\ts + { a_k \beta_k^2 \varphi_k^T P_k \varphi_k + a_k \beta_k^2 \varphi_k^T P_k  \varphi_k \varepsilon_{k+1}^2 }\label{5thm1}.
\eea
By the definition of  $\beta_k$  in Algorithm \ref{algorithm3}, (\ref{psi}), and the Mean Value Theorem, it follows
$$
2 \beta_k \tilde{\theta}_k^T \varphi_k \psi_k = 2 \beta_k (\tilde{\theta}_k^T \varphi_k)^2 f_k(\xi_k)(1 - (p + q)) \geq 2\beta_k^2 (\tilde{\theta}_k^T \varphi_k)^2
$$
where \(\xi_k\) lies between \(C - \theta^T \varphi_k\) and \(C - \hat{\theta}_k^T \varphi_k\). Then by summing up both sides of (\ref{5thm1}), it becomes
\bea
    V_{n+1} \ts\leq\ts  V_0- \sum_{k=0}^{n}\beta_k^2 (\tilde{\theta}_k^T \varphi_k)^2\nn \\
    \ts\ts +  \underbrace{ 2\sum_{k=0}^{n} a_k \beta_k^2 \psi_k \varphi_k^T P_k \varphi_k \varepsilon_{k+1} - 2\sum_{k=0}^{n}  \beta_k   \tilde{\theta}_k^{uT} \varphi_k\varepsilon_{k+1} }_{\text{I}} \nn \\
    \ts\ts + \underbrace{ \sum_{k=0}^{n} a_k \beta_k^2 \varphi_k^T P_k\varphi_k +\sum_{k=0}^{n}  a_k \beta_k^2\varphi_k^T P_k \varphi_k \varepsilon_{k+1}^2 }_{\text{II}}\label{5thm2}.
\eea
From (\ref{vareps}) and (\ref{mean1}), it follows
$
\mathbb{E}(\varepsilon_{k+1} \mid \mathcal{F}_k) = 0,
$
which means $\{\varepsilon_k, \mathcal{F}_k\}$ is a martingale difference sequence.
Similar to the proof in \cite{guo2020time} and noting Assumption A3(a) \( \sup_{k \geq 1} \|\varphi_k\| \leq M < \infty \), we have
\bea
{\rm I}\ts=\ts o\left(\sum_{k=0}^{n}\beta_k^2 (\tilde{\theta}_k^T \varphi_k)^2\right)+O(1)\nn\\
{\rm II}\ts=\ts O \left( \log |P_{n+1}^{-1}| \right).\label{5thm3}
\eea
Combining all terms, we obtain
\[
    \tilde{\theta}_{n+1}^T P_{n+1}^{-1} \tilde{\theta}_{n+1} + \sum_{k=0}^{n} \beta_k^2 (\tilde{\theta}_k^T \varphi_k)^2 = O(\log |P_{n+1}^{-1}|), \quad \text{a.s.}
\]
Finally, since \( \{\beta_k\} \) is a non-increasing sequence, we conclude (\ref{error}).
\end{proof}
\noindent{\it Proof of Theorem \ref{thm:estimation} and \ref{thm:regret}}. Directly obtained from Lemma \ref{lem}.

\noindent{\it Proof of Theorem \ref{thm:tracking}}. By Lemma \ref{lem}, the proof follows similarly to \cite[Theorem 3]{zhang2022identification}.


\end{document}